\documentclass[11pt]{article}

\usepackage[margin=1in]{geometry}
\usepackage{microtype}
\usepackage{graphicx}
\graphicspath{{figures/}}
\usepackage{subcaption}
\usepackage{booktabs}
\usepackage{xcolor}
\usepackage{tikz}
\usetikzlibrary{shapes.geometric, arrows.meta, positioning, fit}
\usepackage{enumitem}
\usepackage{ifthen,amssymb}
\usepackage{aliascnt}
\usepackage{algorithm}
\usepackage{algorithmic}
\usepackage{appendix}
\usepackage{tcolorbox}
\usepackage{multirow}
\usepackage{mathtools}
\usepackage{adjustbox}
\usepackage{amsmath}
\usepackage{amssymb}
\usepackage{amsthm}
\usepackage{natbib}
\usepackage{hyperref}
\usepackage[capitalize,noabbrev]{cleveref}
\usepackage{authblk}

\DeclareMathOperator*{\argmin}{arg\,min}
\DeclareMathOperator*{\argmax}{arg\,max}
\newcommand{\R}{\mathbb{R}}
\newcommand{\E}{\mathbb{E}}
\renewcommand{\O}{\mathcal{O}}

\newcommand{\uerr}[1]{%
  \ensuremath{\textcolor{black!80}{\,{\scriptstyle\pm}\,{\scriptstyle #1}}}%
}

\usepackage{siunitx}
\newcommand{\GiB}{\gibi\byte}

\newcommand{\githubrepo}{\url{https://anonymous.4open.science/r/ZO-Stackelberg}}

\newtheorem{thm}{Theorem}[section]

\newaliascnt{lem}{thm}
\newtheorem{lem}[lem]{Lemma}
\aliascntresetthe{lem}
\crefname{lem}{lemma}{lemmas}
\Crefname{lem}{Lemma}{Lemmas}

\newaliascnt{prop}{thm}
\newtheorem{prop}[prop]{Proposition}
\aliascntresetthe{prop}
\crefname{prop}{proposition}{propositions}
\Crefname{prop}{Proposition}{Propositions}

\newaliascnt{corr}{thm}

\aliascntresetthe{corr}
\crefname{corr}{corollary}{corollaries}
\Crefname{corr}{Corollary}{Corollaries}

\newaliascnt{rmk}{thm}
\newtheorem{rmk}[rmk]{Remark}
\aliascntresetthe{rmk}
\crefname{rmk}{remark}{remarks}
\Crefname{rmk}{Remark}{Remarks}

\newaliascnt{example}{thm}
\newtheorem{example}[example]{Example}
\aliascntresetthe{example}
\crefname{example}{example}{examples}
\Crefname{example}{Example}{Examples}

\newaliascnt{mydef}{thm}
\newtheorem{mydef}[mydef]{Definition}
\aliascntresetthe{mydef}
\crefname{mydef}{definition}{definitions}
\Crefname{mydef}{Definition}{Definitions}

\newaliascnt{assumption}{thm}
\newtheorem{assumption}[assumption]{Assumption}
\aliascntresetthe{assumption}
\crefname{assumption}{assumption}{Assumptions}
\Crefname{assumption}{Assumption}{Assumptions}

\usepackage[textsize=tiny]{todonotes}

\title{Zeroth-Order Stackelberg Control in Combinatorial Congestion Games}
\author[1]{Saeed Masiha$^{\star}$}
\author[2]{Sepehr Elahi$^{\star}$}
\author[1]{Negar Kiyavash}
\author[2]{Patrick Thiran}
\affil[1]{EPFL School of Management of Technology}
\affil[2]{EPFL Department of Computer and Communications Sciences}
\date{}

\begin{document}
\maketitle
\begingroup
\renewcommand{\thefootnote}{$\star$}
\footnotetext{Equal contribution.}
\endgroup
\setcounter{footnote}{0}

% \begin{abstract}
% We study Stackelberg control in combinatorial congestion games, where a leader optimizes parameters (e.g., tolls in a network) constrained by the followers' Wardrop equilibrium response.
% Since equilibrium active sets change abruptly, the leader's objective defines a nonsmooth bilevel problem that defies standard differentiation-based techniques.
% We propose \textsc{ZO-Stackelberg}, an oracle-based framework that treats the lower level as a black box: an inner Frank--Wolfe (FW) solver computes approximate equilibria using linear minimization oracles (LMOs), while an outer zeroth-order method updates parameters without differentiating through these iterates.
% For massive strategy spaces where exact LMOs are intractable, we analyze FW with a subsampled LMO. If the sampled candidate set contains an exact minimizer with probability $\kappa_m$, then the expected error decays as $\mathcal{O}(1/(\kappa_m T))$. We propose stratified sampling as a practical way to avoid a vanishing $\kappa_m$ when the strategies that matter most for the Wardrop equilibrium concentrate in a few dominant combinatorial classes (e.g., short paths).
% We then analyze \textsc{ZO-Stackelberg} with subsampled LMOs, proving convergence to generalized Goldstein stationary points of the hyper-objective of the bilevel problem.
% Experiments on real-world networks demonstrate that our method achieves orders-of-magnitude speedups over a differentiating baseline while converging to follower equilibria.
% \end{abstract}
\begin{abstract}
We study Stackelberg (leader--follower) tuning of network parameters (tolls, capacities, incentives) in combinatorial congestion games, where selfish users choose \emph{discrete} routes (or other combinatorial strategies) and settle at a congestion equilibrium.
The leader minimizes a system-level objective (e.g., total travel time) evaluated at equilibrium, but this objective is typically nonsmooth because the set of used strategies can change abruptly.
We propose \textsc{ZO-Stackelberg}, which couples a projection-free Frank--Wolfe equilibrium solver with a zeroth-order outer update, avoiding differentiation through equilibria.
We prove convergence to generalized Goldstein stationary points of the true equilibrium objective, with explicit dependence on the equilibrium approximation error, and analyze subsampled oracles: if an exact minimizer is sampled with probability $\kappa_m$, then the Frank--Wolfe error decays as $\mathcal{O}(1/(\kappa_m T))$. We also propose stratified sampling as a practical way to avoid a vanishing $\kappa_m$ when the strategies that matter most for the Wardrop equilibrium concentrate in a few dominant combinatorial classes (e.g., short paths).
Experiments on real-world networks demonstrate that our method achieves orders-of-magnitude speedups over a differentiation-based baseline while converging to follower equilibria.
\end{abstract}
\section{Introduction}
\label{sec:intro}
Modern transportation, communication, and logistics systems are often influenced by a decision maker who can modify network parameters, e.g., tolls, capacities, or incentives, to steer how users place demand.
We model this as a \emph{Stackelberg} (leader--follower) interaction: the decision maker (the \emph{leader}) moves first, and a large population of selfish users (the \emph{followers}) subsequently responds \citep{migdalas1995bilevel,korilis1997achieving}.

In many applications, followers choose \emph{discrete} objects such as paths, trees, or schedules rather than smooth flow splits.
This yields \emph{combinatorial congestion games} (CCGs), where the feasible strategy set can be exponentially large and the equilibrium can change abruptly as the leader perturbs parameters \citep{wardrop1952traffic,roughgarden_tardos_2002_selfish_routing,karakostas_kolliopoulos_2009_stackelberg}.
See Appendix~\ref{app:applications} for real-world examples of Stackelberg control in CCGs.

\paragraph{Problem formulation.}
The leader chooses parameters $\theta\in\Theta\subseteq\R^k$ (e.g., tolls, capacities, incentives) that affect follower costs.
Followers are modeled as a nonatomic population of total demand $1$ that selects combinatorial strategies $S$ from a family $\mathcal S$ built from $n$ resources (e.g., paths built from network edges).
Rather than tracking each individual, we represent the population by a distribution $z=(z_S)_{S\in\mathcal S}$, where $z_S$ is the fraction of demand routed on strategy $S$ and $\sum_{S} z_S=1$.
This induces a resource-load vector $y(z)\in[0,1]^n$, where $y_i(z)$ is the total mass of followers that use resource $i$.
We write $\mathcal C$ for the set of feasible load vectors induced by distributions over strategies (formal definitions are given in \Cref{sec:2}).
Given $\theta$, each resource $i$ has a congestion-dependent \emph{cost} that increases with its load; the cost of a strategy $S$ is the sum of the costs of the resources it uses.
Followers are assumed to respond \emph{rationally} by reaching a Wardrop equilibrium: under the induced load, all strategies used with positive mass have minimum cost.
This equilibrium condition admits an equivalent convex-optimization formulation: %there is a convex \emph{potential} function $f(\theta,y)$ (of the load vector $y$) whose minimizer over $\mathcal C$ is exactly the equilibrium load.
there exists a convex \emph{potential} function $f(\theta,y)$ such that its minimization over the set  $\mathcal C$ yields exactly the equilibrium load vector \(y\). We denote this minimizer  by $y^\star(\theta)\in\arg\min_{y\in\mathcal C} f(\theta,y)$.
The leader evaluates the resulting equilibrium via an upper-level objective $F(\theta,y)$ (e.g., social cost) and seeks to minimize the \emph{hyper-objective} $\Phi(\theta)\;:=\;F\bigl(\theta,y^\star(\theta)\bigr)$ over $\Theta$.

% \paragraph{Problem formulation.}
% Let $\theta\in\Theta\subseteq\R^k$ denote the leader-controlled parameter vector (e.g., tolls, capacities, incentives) that affects follower costs.
% Followers choose a combinatorial strategy $S$ from a family $\mathcal S\subseteq 2^{[n]}$ built from $n$ resources (e.g., edges in a network);
% we encode $S$ by its incidence vector $\mathbf 1_S\in\{0,1\}^n$ with $[\mathbf 1_S]_i=1$ if resource $i$ is used.
% The induced load $y\in[0,1]^n$ records the fraction of demand on each resource and ranges over the polytope $\mathcal C:=\mathrm{conv}\{\mathbf 1_S:S\in\mathcal S\}$.
% For a given $\theta$, the Wardrop equilibrium load $y^\star(\theta)$ is characterized as the (unique) minimizer of a convex potential $f(\theta,y)$ (a Beckmann-type objective encoding follower costs) over $\mathcal C$, $y^\star(\theta)\in\arg\min_{y\in\mathcal C} f(\theta,y),$
% and the leader’s Stackelberg objective evaluates an equilibrium through an upper-level loss $F(\theta,y)$.
% The resulting \emph{hyper-objective} is $\Phi(\theta):=F\big(\theta,y^\star(\theta)\big),$ and the goal is to minimize $\Phi$ over $\Theta$.

\paragraph{Main challenge: the hyper-objective can be nonsmooth.}
Even when the primitives $(\theta,y)\mapsto f(\theta,y)$ and $(\theta,y)\mapsto F(\theta,y)$ are smooth, the hyper-objective $\Phi$ need not be.\footnote{This is a standard point in perturbation analysis of optimization problems: smooth problem data can still induce a nonsmooth solution map and value function (e.g., due to active-set changes); see, e.g., \citet{bonnans_shapiro_2000,bonnans_shapiro_1998}.}
The key difficulty is that $\mathcal C$ is a polytope: as $\theta$ varies, the optimizer $y^\star(\theta)$ moves along faces, and small parameter changes can alter which constraints are active, equivalently, which combinatorial strategies are used with positive mass.
These active-set switches create \emph{kinks} in the equilibrium map and hence in $\Phi$ (see \Cref{ex:kinks}).
As a consequence, differentiating through equilibria is delicate in CCGs, and the leader's problem cannot be treated as a standard smooth bilevel program.

\paragraph{Prior work: differentiable equilibrium computation.}
A recent line of work makes the equilibrium computation differentiable and then updates the leader parameters $\theta$ by backpropagating through the lower-level iterations \citep{Sakaue2021diff}.
Concretely, they run an iterative equilibrium solver for a fixed number of steps and differentiate through this computation to obtain a hyper-gradient that optimizes the resulting \emph{surrogate} objective evaluated at the approximate equilibrium.
However, convergence guarantees for optimizing this surrogate do not necessarily imply convergence guarantees for the true hyper-objective $\Phi(\theta)=F(\theta,y^\star(\theta))$, especially when $\Phi$ is nonsmooth due to active-set changes, as can happen in CCGs.
To turn surrogate guarantees into guarantees for $\Phi$, one must control both the equilibrium approximation error and the bias introduced by the differentiable approximation, as noted in \citet{Sakaue2021diff}.
A more thorough discussion of related work appears in Appendix~\ref{app:related_work}.

\paragraph{Our approach: oracle-based optimization of the true (nonsmooth) objective.}
We take an alternative viewpoint: rather than enforcing differentiability of $y^\star(\theta)$, we treat equilibrium computation as a black box and optimize $\Phi$ without differentiating through the lower-level solver.
Our method has two coupled components:
(i) an \emph{inner} Frank--Wolfe routine approximates $y^\star(\theta)$ over $\mathcal C$ using only gradients $\nabla_y f(\theta,\cdot)$ and an \emph{exact or subsampled} linear minimization oracle (LMO) over $\mathcal C$;
and (ii) an \emph{outer} zeroth-order loop that updates $\theta$ using two-point function evaluations of
$\widehat\Phi_T(\theta):=F(\theta,y_T(\theta))$, where $y_T(\theta)$ is the inner iterate after $T$ steps.

Our approach has both practical and theoretical advantages.
\emph{Practically}, the inner Frank--Wolfe loop never enumerates the (potentially huge) strategy set $\mathcal S$: each step only needs a linear minimization oracle (LMO) that, given weights $g\in\R^n$, returns a minimum-cost strategy in $\argmin_{S\in\mathcal S}\sum_{i\in S} g_i$ \citep{jaggi2013revisiting}.
In routing games this is a shortest-path computation (e.g., Dijkstra's algorithm) \citep{dijkstra_1959}; for more complex strategy families it can be implemented via combinatorial routines or decision-diagram methods \citep{minato_1993_zdd,Sakaue2021diff}.
When exact oracle calls are too expensive, we use a \emph{subsampled} LMO that minimizes over $m$ sampled candidate strategies \citep{kerdreux2018fw_subsampling}.
\emph{Theoretically}, the resulting bilevel method comes with convergence guarantees.

In terms of theoretical guarantees, first, for subsampled Frank--Wolfe, we isolate a single optimizer-hit parameter $\kappa_m$, the probability that the sampled candidate set contains an exact LMO minimizer, and prove an $\O(1/(\kappa_m T))$ convergence rate.
This recovers the standard subsampled Frank--Wolfe guarantee under uniform-inclusion sampling, i.e., each strategy has the same probability of appearing in the candidate set (for example, by sampling $m$ strategies uniformly at random) \citep[Theorem~2.1]{kerdreux2018fw_subsampling}, while allowing general (non-uniform, structure-aware) sampling distributions quantified solely by $\kappa_m$.
Second, we prove that the full bilevel procedure converges to a generalized Goldstein stationary point of the hyper-objective $\Phi$ \citep{goldstein1977optimization}, with an explicit dependence on the inner equilibrium approximation error.

\paragraph{Contributions.}
Our main contributions are as follows.
\begin{itemize}[leftmargin=0pt, topsep=0pt, partopsep=0pt, parsep=0pt, itemsep=1pt]
\item Since the hyper-objective $\Phi(\theta)=F(\theta,y^\star(\theta))$ is typically nonsmooth in CCGs, assuming it is differentiable is unrealistic. We therefore propose \textsc{ZO-Stackelberg}, an oracle-based bilevel algorithm that targets the original $\Phi$ without differentiating through equilibria by coupling an inner Frank--Wolfe equilibrium solver with an outer zeroth-order update on~$\theta$.
\item We establish an end-to-end convergence guarantee for \textsc{ZO-Stackelberg}: the method converges to generalized Goldstein stationary points (GGSPs) of the (Lipschitz, possibly nonsmooth) hyper-objective $\Phi$, with explicit dependence on the inner equilibrium approximation error.
\item To make the inner equilibrium solver scalable when exact LMOs are expensive, we use and analyze subsampled Frank--Wolfe and prove an $\O(1/(\kappa_m T))$ rate under a mild optimizer-hit assumption, where $\kappa_m$ is the probability that the sampled candidate set contains an exact LMO minimizer. We also propose stratified (e.g., length-debiased) sampling schemes that keep $\kappa_m$ nontrivial in large strategy spaces.
\item We provide an efficient Python implementation with both exact and subsampled LMOs (including $s$--$t$ path, Hamiltonian path, and Steiner cycle strategy families). Experiments on real-world transportation networks show that \textsc{ZO-Stackelberg} matches the accuracy of a differentiation-based baseline while achieving orders-of-magnitude speedups with drastically reduced memory usage.
\end{itemize}

% \paragraph{Contributions.}
% To summarize, our main contributions are as follows.
% \begin{itemize}[leftmargin=0pt, topsep=0pt, partopsep=0pt, parsep=0pt, itemsep=1pt]
% \item We propose an oracle-based bilevel algorithm for Stackelberg CCGs that combines a Frank--Wolfe equilibrium solver with a zeroth-order outer loop, and does \emph{not} require differentiating through equilibria.
% \item Our subsampled Frank--Wolfe guarantees an $\O(1/(\kappa_m T))$ rate under a mild optimizer-hit assumption, where $\kappa_m$ is the probability that the sampled candidate set contains an exact LMO minimizer.
% We also propose stratified (e.g., length-debiased) sampling schemes that keep $\kappa_m$ nontrivial in large, imbalanced strategy spaces (Section~\ref{sec:inner_subsampled}).
% \item  We prove that our zeroth-order method converges to generalized Goldstein stationary points (GGSPs) of the (Lipschitz, possibly nonsmooth) hyper-objective $\Phi$, with explicit dependence on the inner equilibrium approximation error.
% \item We provide efficient Python implementations of our zeroth-order method with both exact and subsampled LMOs based on ZDDs, including for st-path, Hamiltonian path, and Steiner cycle strategy families. We also provide experiments on real-world transportation networks that highlight our method's flexibility and performance (\Cref{sec:experiments}).
% \end{itemize}

The rest of the paper is organized as follows.
\Cref{sec:2} formalizes the Stackelberg model of CCGs.
\Cref{sec:zeroth-order} presents the bilevel algorithm.
\Cref{sec:lmo_imp} discusses implementations of the LMO.
\Cref{sec:theory} provides the convergence analysis, and \Cref{sec:experiments} reports empirical results.

\section{Problem Setting}\label{sec:2}

We formalize the Stackelberg CCG model and state the assumptions used in our analysis.
A summary of notation appears in Appendix~\ref{app:notation}.

\subsection{Combinatorial Congestion Games (CCGs)}
% A unit-mass nonatomic population chooses strategies $S\in\mathcal S\subseteq 2^{[n]}$ over resources $[n]:=\{1,\dots,n\}$.
Let $[n]:=\{1,\dots,n\}$ index the resources (e.g., edges) and let $\mathcal S\subseteq 2^{[n]}$ be the set of feasible strategies.
We represent a nonatomic population by a distribution $z=(z_S)_{S\in\mathcal S}\in\Delta^d$ over strategies, where $d:=|\mathcal S|$ and
\(\Delta^d:=\{z\in\R^d : z_S\ge 0,\ \sum_{S\in\mathcal S} z_S = 1\}\).
We encode $S$ by its incidence vector $\mathbf 1_S\in\{0,1\}^n$ with $[\mathbf 1_S]_i=1$ if resource $i$ is used. Then the induced resource-load vector is
\(y(z):=\sum_{S\in\mathcal S} z_S\,\mathbf 1_S\),
and the feasible set of loads is the polytope $\mathcal C:=\mathrm{conv}\{\mathbf 1_S:S\in\mathcal S\}\subseteq[0,1]^n$.

\paragraph{Cost functions.}
Each resource $i\in[n]$ has a congestion-dependent cost function $c_i:[0,1]\to\R$, and the cost of a strategy $S$ under load $y\in\mathcal C$ is $c_S(y):=\sum_{i\in S} c_i(y_i)$.

\begin{assumption}[Resource costs]
\label{ass:costs}
For each $i \in [n]$, the function $c_i$ is continuous and strictly increasing on $[0,1]$.
\end{assumption}

% \subsection{Wardrop equilibrium and potential minimization}
% We use Wardrop equilibrium \citep{wardrop1952traffic}: under the induced load, every strategy used with positive mass has minimum cost. For now we omit the parameter~$\theta$; the parameterized Stackelberg model is stated in Section~\ref{sec:2.3}.

\subsection{Wardrop Equilibrium and Potential Minimization}
We adopt the notion of Wardrop equilibrium \citep{wardrop1952traffic}, which is the standard equilibrium concept for nonatomic congestion games, under which, for a given induced load, every strategy used with positive mass has minimum cost.
For ease of notation, we temporarily omit the parameter~$\theta$; the parameterized Stackelberg model is introduced in Section~\ref{sec:2.3}.

% We now define Wardrop equilibrium and introduce the associated potential function.
\begin{mydef}[Wardrop equilibrium]
A population profile $z \in \Delta^d$ with induced load $y = y(z) \in \mathcal{C}$ is a
\emph{Wardrop equilibrium} if the following holds:
for every strategy $S \in \mathcal{S}$ with $z_S > 0$, we have $c_S(y)\le c_{S'}(y)$ for all $S' \in \mathcal{S}$.
% In words, all strategies that are used with positive mass have minimal cost, so no player
% can decrease her cost by unilaterally switching to another feasible strategy.
\end{mydef}

% We associate to the game the standard convex potential \citep{beckmann1956studies} $f(y):=\sum_{i\in[n]}\int_0^{y_i} c_i(u)\,\text{d}u$.
% Under Assumption~\ref{ass:costs}, $f$ is continuously differentiable and strictly convex on $[0,1]^n$, with $\nabla f(y)_i=c_i(y_i)$.

We define the standard convex potential of the game \citep{beckmann1956studies} as $f(y):=\sum_{i\in[n]}\int_0^{y_i} c_i(u)\,\text{d}u$.
Under Assumption~\ref{ass:costs}, $f$ is continuously differentiable and strictly convex on $[0,1]^n$, with $\nabla f(y)_i=c_i(y_i)$. As a result, $f$ admits a unique minimizer over the convec set $\mathcal{C}$.

% We consider the convex optimization problem $\min_{y \in \mathcal{C}} f(y)$.
The next proposition recalls the classical equivalence between Wardrop equilibria and minimizers of this problem, originating with the Beckmann transformation \citep{beckmann1956studies}; see also standard references such as \citet[Section~3.2]{sheffi1984urban}.
We provide a proof in \Cref{append_prop:equilibrium-potential} for completeness.
\begin{prop}[Equilibrium $\iff$ potential minimizer]
\label{prop:equilibrium-potential}
Let $z \in \Delta^d$ and $y = y(z)$. Under Assumption~\ref{ass:costs}, $z$ is a Wardrop equilibrium if and only if $y =\argmin_{y'\in\mathcal{C}}f(y')$.
% In particular, a load vector $y \in \mathcal{C}$ arises at some Wardrop equilibrium if and
% only if it minimizes the potential $f$ over $\mathcal{C}$.
\end{prop}
% Proposition~\ref{prop:equilibrium-potential} allows us to identify Wardrop equilibria with the
% (unique) minimizer $y^\star$ of the convex program~\eqref{eq:potential-minimization}. 
In what follows we will therefore work primarily with the load vector~$y$.

\subsection{Bilevel Model of CCG}\label{sec:2.3}
The leader controls a vector $\theta \in \Theta \subseteq \mathbb{R}^k$ (closed and convex) that parametrizes the resource cost functions.
For each resource $i\in[n]$, let $c_i(\cdot;\theta):[0,1]\to\R$ satisfy Assumption~\ref{ass:costs} for every fixed $\theta\in\Theta$, and assume $\theta\mapsto c_i(y_i;\theta)$ is continuously differentiable for every fixed $y_i\in[0,1]$.
The associated potential is $f(\theta,y):=\sum_{i\in[n]}\int_0^{y_i} c_i(u;\theta)\,du$, where $y\in\mathcal C$.
For each $\theta\in\Theta$, define the equilibrium load $y^\star(\theta):=\arg\min_{y\in\mathcal C} f(\theta,y)$.
The leader evaluates an equilibrium via an upper-level objective $F:\Theta\times\mathcal C\to\R$ and solves the Stackelberg problem
\begin{equation}\label{eq:stackelberg-problem}
  \min_{\theta\in\Theta} F(\theta,y)
  \quad\text{s.t.}\quad
  y\in\arg\min_{u\in\mathcal C} f(\theta,u).
\end{equation}
Since the lower-level minimizer is unique, we write $y^\star(\theta)$ and optimize the reduced objective $\Phi(\theta):=F(\theta,y^\star(\theta))$ over $\Theta$.
% where $\theta\mapsto y^\star(\theta)$ is the \emph{equilibrium map} of a parametric convex program.

\paragraph{Lipschitz stability of the equilibrium map.}
As discussed in \Cref{sec:intro}, $\Phi$ can be nonsmooth because $y^\star(\theta)$ may switch between faces of the polytope $\mathcal C$.
Even in this nonsmooth regime, the equilibrium map $\theta\mapsto y^\star(\theta)$ (and thus the hyper-objective $\Phi$) can still be stable in a weaker sense: under mild conditions, both are globally Lipschitz.
To formalize this, we assume a local quadratic growth property around each equilibrium, together with standard Lipschitz regularity of $\nabla_y f$ and of $F$.

\begin{assumption}[Local quadratic growth around equilibrium]\label{ass:lsm}
There exist constants $\alpha>0$ and $r>0$ such that for all $\theta\in\Theta$ and all $y\in \mathcal{C}$ with $\|y-y^\star(\theta)\|\le r$, we have
\begin{equation}\label{eq:qg}
f(\theta,y)-f(\theta,y^\star(\theta)) \;\ge\; \frac{\alpha}{2}\,\|y-y^\star(\theta)\|^2
\end{equation}
\end{assumption}

\begin{assumption}[Regularity of the lower- and upper-level problems]\label{ass:LL_UL_regular}
There exists $L_{f,2}\ge 0$ such that for all $\theta,\theta'\in\Theta$ and all $y,y'\in \mathcal{C}$,
\(\bigl\|\nabla_y f(\theta,y)-\nabla_y f(\theta',y')\bigr\| \le L_{f,2}\bigl(\|\theta-\theta'\|+\|y-y'\|\bigr)\).
Moreover, $F:\Theta\times\mathcal C\to\R$ is Lipschitz: there exists
$L_{F,1}\ge 0$ such that for all $\theta,\theta'\in\Theta$ and all $y,y'\in \mathcal{C}$,
\(|F(\theta,y)-F(\theta',y')| \le L_{F,1}[\|\theta-\theta'\| + \|y-y'\|]\).
\end{assumption}
See Appendix~\ref{app:constant_orders} for simple bounds and typical scaling of $L_{f,2}$ and $L_{F,1}$ in our setting.
Under \Cref{ass:lsm,ass:LL_UL_regular}, the equilibrium map is stable to perturbations of $\theta$; Lemma~\ref{lem_solution_map_Lip} makes this precise (proof in Appendix~\ref{app:differentiability}).

\begin{lem}[Lipschitzness of the equilibrium map and hyper-objective]\label{lem_solution_map_Lip}
Let \Cref{ass:lsm,ass:LL_UL_regular} hold and assume that $\Theta$ is convex.
Then the equilibrium map $\theta\mapsto y^\star(\theta)$ and the hyper-objective $\Phi(\theta)=F(\theta,y^\star(\theta))$ are globally Lipschitz on $\Theta$:
\begin{equation}\label{eq:y_star_global_lip}
\|y^\star(\theta)-y^\star(\theta')\| \le \frac{L_{f,2}}{\alpha}\,\|\theta-\theta'\|
\qquad\forall \theta,\theta'\in\Theta,
\end{equation}
and %Moreover, the hyper-objective $\Phi(\theta)=F(\theta,y^\star(\theta))$ is globally Lipschitz on $\Theta$:
\begin{equation}\label{eq:phi_global_lip}
|\Phi(\theta)-\Phi(\theta')|
\le L_\Phi\|\theta-\theta'\|
\quad\forall \theta,\theta'\in\Theta,
\end{equation}
where $L_\Phi:=L_{F,1}+{L_{F,1}L_{f,2}}/{\alpha}$.
\end{lem}

\paragraph{When does quadratic growth hold?}
Assumption~\ref{ass:lsm} is implied by (local) strong convexity of $y\mapsto f(\theta,y)$ over $\mathcal C$.
For the separable potential $f(\theta,y)=\sum_i\int_0^{y_i}c_i(u;\theta)\,du$, a simple sufficient condition is a uniform slope lower bound:
if each $c_i(\cdot;\theta)$ is differentiable in its first argument and $\partial_{y}c_i(y_i;\theta)\ge \mu>0$ for all $y_i\in[0,1]$ and all $\theta\in\Theta$,
then $f(\theta,\cdot)$ is $\mu$-strongly convex (hence~\eqref{eq:qg} holds globally, with $r=\infty$ and $\alpha=\mu$).
This is the case for the affine-in-load costs used in our experiments (\Cref{sec:experiments}).

\paragraph{Nonsmoothness via active-set changes.}
We illustrate the kink phenomenon from \Cref{sec:intro} with a two-link example.

\begin{example}[Kinks from active-set changes]\label{ex:kinks}
Consider a parallel network with two resources $i\in \{1,2\}$, strategies $\mathcal S=\{\{1\},\{2\}\}$, and demand $1$.
Let $c_1(y_1;\theta)=y_1$ and $c_2(y_2;\theta)=\theta$ parametrized by $\theta \in \Theta = \R$.
Then
\[
  f(\theta,y)=\tfrac12 y_1^2+\theta y_2,
  \qquad y_1+y_2=1,\;\; y_1,y_2\in[0,1].
	\]
	Eliminating $y_2=1-y_1$ yields $\min_{y_1\in[0,1]} \tfrac12 y_1^2+\theta(1-y_1)$, which is minimized at
	$y_1^\star(\theta)=\min\{1,\max\{0,\theta\}\}$ and $y_2^\star(\theta)=1-y_1^\star(\theta)$.
	Thus $y^\star(\theta)$ is $1$-Lipschitz continuous but it is not differentiable at $\theta=0$ and $\theta=1$, where the
	active constraints (and the set of resources in use) change. The same active-set switching can occur many times as the strategy set grows; see \Cref{append_ex:many_kinks} for a scalable example with $\Omega(|\mathcal S|)$ kinks.
\end{example}
Example~\ref{ex:kinks} shows that $\Phi$ can have kinks when the equilibrium changes its active constraints, so gradients may not exist. We therefore use a standard zeroth-order approach to circumvent the issue.

\section{Zeroth-Order Algorithm for Stackelberg Control}
\label{sec:zeroth-order}

Let us describe our oracle-based procedure for solving $\min_{\theta\in\Theta}\Phi(\theta)=F(\theta,y^\star(\theta))$.
The method alternates between (i) an inner equilibrium solver that returns an approximate equilibrium $y_T(\theta)$ and (ii) an outer zeroth-order update that uses function values of $\widehat\Phi_T(\theta):=F(\theta,y_T(\theta))$.

\subsection{Inner Loop: Equilibrium Approximation via Frank--Wolfe}
\label{sec:inner_fw_algo}

Given $\theta$,  we solve the inner problem, i.e.,  $\min_{y\in\mathcal C} f(\theta,y)$, approximately using Frank--Wolfe (FW).
At iterate $y_t$, FW forms $g_t:=\nabla_y f(\theta,y_t)$ and calls a linear minimization oracle (LMO) over $\mathcal C$:
\(
\mathrm{LMO}(g)\in\arg\min_{y\in\mathcal C}\langle g,y\rangle
=\arg\min_{S\in\mathcal S}\sum_{i\in S} g_i.
\)
We use either the exact LMO or a subsampled variant $\mathrm{LMO}_m$ (Section~\ref{sec:lmo_imp}). \Cref{alg:fw-inner} summarizes the inner loop.

\begin{algorithm}[t]
  \caption{\textsc{FW-Equilibrium}$(\theta,T)$: Frank--Wolfe equilibrium solver}
  \label{alg:fw-inner}
  \begin{algorithmic}[1]
    \STATE \textbf{Input:} parameter $\theta\in\Theta$, iterations $T\ge 1$
    \STATE initialize $y_0\in\mathcal C$ (e.g., $y_0=1_{S_0}$ for some $S_0\in\mathcal S$)
    \FOR{$t=0,1,\dots,T-1$}
      \STATE compute $g_t := \nabla_y f(\theta,y_t)$
      \STATE oracle call $s_t := \mathrm{LMO}(g_t)$ \hfill (or $s_t := \mathrm{LMO}_m(g_t)$)
      \STATE choose $\gamma_t\in\arg\min_{\gamma\in[0,1]} f\big(\theta, y_t+\gamma(s_t-y_t)\big)$ \hfill (line search)
      \STATE update $y_{t+1} := (1-\gamma_t)y_t + \gamma_t s_t$
    \ENDFOR
    \STATE \textbf{return} $y_T(\theta):=y_T$
  \end{algorithmic}
\end{algorithm}

We use an exact line search along the FW segment.
The output $y_T(\theta)$ defines the approximate hyper-objective value $\widehat\Phi_T(\theta):=F(\theta,y_T(\theta))$.
We quantify the inner accuracy as a function of $T$ in \Cref{sec:theory}.

\subsection{Outer Loop: Zeroth-Order Updates on $\Theta$}
\label{sec:outer_zo_algo}

The outer loop performs projected zeroth-order updates on $\Theta$ using two-point finite differences.
At iteration $t$, sample $u_{t,1},\dots,u_{t,B}$ uniformly from the unit sphere, form the estimator\footnote{If $F$ is only defined on $\Theta$, then evaluating $F(\theta_t\pm\rho u_{t,i},\cdot)$ can be ill-posed when $\theta_t\pm\rho u_{t,i}\notin\Theta$. A standard workaround is to use an \emph{interiorized} feasible set $\Theta_\rho:=\{\theta\in\Theta:\theta+\rho\mathbb B\subseteq\Theta\}$ and query from $\bar\theta_t:=\Pi_{\Theta_\rho}(\theta_t)$, so that $\bar\theta_t\pm\rho u_{t,i}\in\Theta$.}
\(\widehat g_t:=\frac{k}{2\rho B}\sum_{i=1}^B\big(\widehat\Phi_T(\theta_t+\rho u_{t,i})-\widehat\Phi_T(\theta_t-\rho u_{t,i})\big)u_{t,i}\),
and update $\theta_{t+1}:=\Pi_\Theta(\theta_t-\eta\,\widehat g_t)$.
Each outer iteration makes $2B$ calls to the inner solver (and hence $2BT$ LMO calls). Algorithm~\ref{alg:zo-outer} returns the last iterate $\theta_K$; for analysis we consider a uniformly random iterate from $\{\theta_0,\dots,\theta_{K-1}\}$ (Section~\ref{sec:theory}).
% Each function evaluation $\widehat\Phi_T(\cdot)$ is obtained by running the inner solver for $T$ iterations to compute an
% approximate equilibrium $y_T(\cdot)$.
\begin{algorithm}[t]
  \caption{\textsc{ZO-Stackelberg}: zeroth-order Stackelberg optimization}
  \label{alg:zo-outer}
  \begin{algorithmic}[1]
    \STATE \textbf{Input:} initial $\theta_0\in\Theta$, outer iterations $K$, inner iterations $T$, smoothing $\rho>0$, stepsize $\eta>0$, batch size $B\in\mathbb N$
    \FOR{$t=0,1,\dots,K-1$}
      \STATE sample $u_{t,1},\dots,u_{t,B}$ i.i.d.\ uniformly from the unit sphere in $\R^k$
      \STATE $\widehat g_t \gets 0$
      \FOR{$i=1,\dots,B$}
        \STATE $y^+_{t,i} := \textsc{FW-Equilibrium}(\theta_t+\rho u_{t,i},T)$
        \STATE $y^-_{t,i} := \textsc{FW-Equilibrium}(\theta_t-\rho u_{t,i},T)$
        \STATE $\widehat\Phi^+_{t,i} := F(\theta_t+\rho u_{t,i},y^+_{t,i})$
        \STATE $\widehat\Phi^-_{t,i} := F(\theta_t-\rho u_{t,i},y^-_{t,i})$
        \STATE $\widehat g_t \gets \widehat g_t + \frac{k}{2\rho}\,(\widehat\Phi^+_{t,i}-\widehat\Phi^-_{t,i})\,u_{t,i}$
      \ENDFOR
      \STATE $\widehat g_t := \frac{1}{B}\,\widehat g_t$
      \STATE $\theta_{t+1} := \Pi_\Theta\big(\theta_t-\eta\,\widehat g_t\big)$
    \ENDFOR
    \STATE \textbf{return} $\theta_K$
  \end{algorithmic}
\end{algorithm}

\section{Implementing the Linear Minimization Oracle (LMO)}
\label{sec:lmo_imp}

Each Frank--Wolfe step in \Cref{sec:inner_fw_algo} requires a \emph{linear minimization oracle} (LMO): given weights $g\in\R^n$, the oracle returns $\mathrm{LMO}(g)\in\arg\min_{S\in\mathcal S}\sum_{i\in S} g_i$.
In our congestion setting, FW uses $g_t=\nabla_y f(\theta,y_t)$, so $[g_t]_i=c_i(y_{t,i};\theta)$.
The computational nature of the oracle is therefore dictated by the combinatorial family $\mathcal S$.
In graph-based CCGs, resources are edges (or nodes) and strategies are discrete structures such as $s$--$t$ paths, Steiner trees, or Hamiltonian cycles \citep{roughgarden_tardos_2002_selfish_routing,anshelevich_et_al_2004_price_of_stability,garey_johnson_tarjan_1976_planar_hamiltonian}.
Accordingly, we split the discussion by the complexity of the underlying family: \Cref{sec:lmo_poly} covers polynomial-time LMOs for tractable $\mathcal S$, while \Cref{sec:lmo_zdd} develops exact and approximate LMOs for NP-hard families.

\subsection{Polynomial-Time LMO for Tractable Families}
\label{sec:lmo_poly} 
When the minimum-weight feasible strategy can be found in polynomial time, FW calls that routine directly and never enumerates $\mathcal S$.
For example, in $s$--$t$ routing on a directed graph $G=(V,E)$, the LMO is a shortest-path problem $\min_{P\in\mathcal P_{s\to t}}\sum_{e\in P} g_e$, solvable by Dijkstra's algorithm when weights are nonnegative \citep{dijkstra_1959}.
Other tractable families are handled analogously.

\subsection{Exact \& Approximate LMO for NP-hard Families}
\label{sec:lmo_zdd}

For NP-hard families, the LMO is itself NP-hard as a standalone combinatorial optimization problem.
However, across FW iterations (and across outer iterations), the feasible family $\mathcal S$ remains fixed and only the additive weights $g$ vary.
We leverage this structure by compiling $\mathcal S$ once into a \emph{zero-suppressed binary decision diagram} (ZDD) \citep{minato_1993_zdd}.
A ZDD is a directed acyclic graph (DAG) that compactly represents a family of subsets (i.e., $\mathcal S$) by sharing identical subproblems, often making it far smaller than the explicit strategy set.
In our setting, each root-to-$\top$ path encodes a feasible strategy, where $\top$ denotes the accepting terminal (see \Cref{app:zdd} for notation and a primer).
Although compilation can be exponential in $n$ in the worst case, it amortizes the cost of NP-hardness and enables many subsequent LMO queries without re-solving the problem from scratch.

\paragraph{Exact LMO via dynamic programming on the ZDD.}
Given weights $g\in\R^n$, assign cost $g_i$ whenever item $i$ is included in the strategy (and $0$ otherwise).
Then the LMO is a minimum-cost root-to-$\top$ path in the ZDD DAG, which is solved by a single bottom-up dynamic program and a traceback.
This yields an \emph{exact} LMO in time linear in the ZDD size (see Appendix~\ref{app:lmo_complexity}).
ZDD size can be exponential in the worst case, but in our instances it is typically manageable (\Cref{sec:experiments}).

\paragraph{Approximate LMO via Subsampling.}
\label{sec:lmo_subsample}
When the ZDD is too large for repeated exact minimization, we use it instead to \emph{sample} feasible strategies.
Since root-to-$\top$ paths correspond to strategies, sampling paths in the ZDD produces samples from $\mathcal S$ without enumerating $\mathcal S$.
Given weights $g\in\R^n$, we sample $m$ i.i.d. feasible strategies $S^{(1)},\dots,S^{(m)}\in\mathcal S$, score each by $\sum_{i\in S} g_i$, and return the best one(s) as a subsampled oracle $\mathrm{LMO}_m$.
We discuss sampling distributions and analyze the resulting subsampled FW rate in \Cref{sec:inner_subsampled}, with implementation details in \Cref{app:zdd_sampling}.
In our experiments, $\mathrm{LMO}_m$ closely tracks the exact ZDD oracle while remaining practical when repeated exact minimization is not (\Cref{sec:experiments}).
Appendix~\ref{app:lmo_complexity} summarizes preprocessing and per-iteration costs.

Unlike \citet{Sakaue2021diff}, who smooth the LMO to backpropagate through FW iterates, we use ZDDs only as exact or sampled combinatorial oracles inside standard FW.
The outer loop then optimizes the true (typically nonsmooth) hyper-objective without differentiating through equilibria.

\section{Theoretical Analysis}\label{sec:theory}

We analyze \Cref{alg:zo-outer} for minimizing $\Phi(\theta)=F(\theta,y^\star(\theta))$, where $y^\star(\theta)\in\arg\min_{y\in\mathcal C} f(\theta,y)$, using either an exact or subsampled LMO.

\subsection{Inner Loop: Equilibrium Approximation via (Subsampled) Frank--Wolfe}
\label{sec:inner_subsampled}

Given $\theta$, \Cref{alg:fw-inner} runs FW on $\min_{y\in\mathcal C} f(\theta,y)$, where $\mathcal C=\mathrm{conv}\{\mathbf 1_S:S\in\mathcal S\}$.
Each FW step calls an LMO over $\mathcal C$; let $D:=\max\{\|u-v\|:u,v\in\mathcal C\}$.

\paragraph{Exact and subsampled linear minimization.}
Given $g\in\R^n$, the LMO returns
\(\mathrm{LMO}(g)\in\arg\min_{y\in\mathcal C}\langle g,y\rangle=\arg\min_{S\in\mathcal S}\langle g,\mathbf 1_S\rangle\).
When $\mathcal S$ is too large, we use a subsampled oracle: draw a candidate set $S^{(1)},\dots,S^{(m)}\overset{\text{i.i.d.}}{\sim}q$ and set
\(\mathrm{LMO}_m(g)\in\arg\min_{1\le j\le m}\langle g,\mathbf 1_{S^{(j)}}\rangle\).
The inner loop is standard FW with $\mathrm{LMO}$ replaced by $\mathrm{LMO}_m$.

\begin{rmk}[Exact LMO (standard)]\label{rmk:fw_exact_lmo}
If the exact oracle $\mathrm{LMO}$ is available, FW admits the standard $\O(1/T)$ convergence rate (e.g., \citet{frank1956algorithm,jaggi2013revisiting}).
We therefore focus below on the subsampled oracle.
\end{rmk}

\paragraph{The single quantity that controls the effect of subsampling.}
Let $\mathrm{Opt}(g):=\arg\min_{S\in\mathcal S}\langle g,\mathbf 1_S\rangle$ and $p(g):=q(\mathrm{Opt}(g))$.
With $m$ i.i.d.\ samples, the probability of sampling an exact minimizer is $\kappa_m(g):=1-(1-p(g))^m$.
For a clean FW rate we assume that $p(g)$ is uniformly bounded from below along the inner run:
\begin{assumption}[Uniform optimizer mass]\label{ass:opt_mass}
There exists $\underline p\in(0,1]$ such that for all FW gradients $g$ encountered by the inner loop,
\(p(g)=q(\mathrm{Opt}(g))\ge \underline p\).
\end{assumption}

Under \Cref{ass:opt_mass}, $\kappa_m(g)\ge \kappa_m:=1-(1-\underline p)^m$, and $\kappa_m$ is the only dependence on $q$ in the FW rate.
In the bilevel setting, we require this bound uniformly over the inner FW runs invoked by the outer loop, so $\kappa_m$ does not depend on $\theta$.

\begin{rmk}[Relation to uniform-inclusion subsampling]\label{rmk:opt_mass_vs_uniform_inclusion}
\citet{kerdreux2018fw_subsampling} analyzed a subsampled FW scheme under a \emph{uniform-inclusion} model, where each atom is included in the candidate set with the same probability $\eta$, yielding an $\mathcal O(1/(\eta T))$ rate.
In this setting, the probability that the candidate set contains an exact LMO minimizer is at least $\eta$.
\Cref{ass:opt_mass} weakens this condition by requiring that, under a general (possibly non-uniform) sampling distribution $q$, the probability of sampling an exact LMO minimizer is bounded below by $\kappa_m$.
The uniform-inclusion model is therefore a special case with ($\kappa_m=\eta$), while our assumption allows for non-uniform, structure-based sampling schemes.
\end{rmk}

\begin{thm}[Convergence of \textsc{FW-Equilibrium} with subsampled LMO]\label{thm:fw_equilibrium_subsampled}
Let \Cref{ass:costs,ass:LL_UL_regular,ass:opt_mass} hold. Fix $\theta\in\Theta$ and initialize $y_0\in\mathcal C$.
Run \Cref{alg:fw-inner} with the subsampled oracle $\mathrm{LMO}_m$ and exact line-search.\footnote{Appendix~\ref{app:proof_fw_subsampled} shows the same bound holds with the standard smoothness-based short-step (approximate line search) $\gamma_t=\min\{\max\{\langle \nabla_y f(\theta,y_t),y_t-s_t\rangle/(L_{f,2}D^2),0\},\,1\}$.}
Define the initial optimality gap $\Delta_0 := f(\theta,y_0)-f(\theta,y^\star(\theta))$.
Then for all $T\ge T_0$,
\begin{equation}\label{eq:fw_linesearch_rate_kappa_burnin}
\E\!\left[f(\theta,y_T)\right]-f(\theta,y^\star(\theta))
\;\le\;
\frac{2L_{f,2}D^{2}}{\kappa_m\,(T-T_0+1)},
\end{equation}
where
\begin{equation}\label{eq:burn_in_def}
T_0:=\left\lceil \frac{1}{\kappa_m}\log\Bigl(1+\frac{2\Delta_0}{L_{f,2}D^{2}}\Bigr)\right\rceil .
\end{equation}
Moreover, if \Cref{ass:lsm} holds, then for every $\varepsilon\in(0,r]$,
\(\mathbb{E}[\|y_T-y^\star(\theta)\|^{2}]\le \varepsilon^{2}\) whenever
\(T \ge \left\lceil 4L_{f,2}D^2\,\max\{1,D/r\}/(\alpha\,\kappa_m\,\varepsilon^2)-1\right\rceil + T_0\),
where $T_0$ is defined in \eqref{eq:burn_in_def}.
If $\bar\Delta_0:=\sup_{\theta\in\Theta}\Delta_0(\theta)<\infty$\footnote{This uniform boundedness holds automatically when $\Theta$ is compact, since $f$ and $y^{*}(\cdot)$ are continuous on the compact set $\Theta\times\mathcal C$.} and $\bar T_0$ is defined by \eqref{eq:burn_in_def} with $\Delta_0$ replaced by $\bar\Delta_0$, then the same bound holds uniformly with $\bar T_0$ in place of $T_0$.
\end{thm}
The proof is given in Appendix~\ref{app:proof_fw_subsampled}.
The rate \eqref{eq:fw_linesearch_rate_kappa_burnin} is the standard $\O(1/T)$ FW rate up to the logarithmic burn-in $T_0$, with subsampling captured by $\kappa_m$; we next discuss how to make $\kappa_m$ nontrivial.

\paragraph{Our idea to improve $p(g)$: stratified sampling.}
If $q$ is uniform over $\mathcal S$, then
$p(g)=|\mathrm{Opt}(g)|/|\mathcal S|$. When $|\mathcal S|$ is huge (often exponential in $n$),
$p(g)$ is tiny unless $\mathrm{Opt}(g)$ is enormous, and for $m\ll|\mathcal S|$ we have $\kappa_m(g)=1-(1-p(g))^m\approx m\,p(g)$, so the bound in \eqref{eq:fw_linesearch_rate_kappa_burnin} is vacuous unless $m$ is extremely large.

Uniform sampling can be ineffective in \emph{imbalanced} families, where the number of feasible strategies varies widely across simple strata (e.g., by cardinality $|S|$).
In such cases, uniform draws are dominated by the largest strata, even though LMO minimizers often come from more structured classes (e.g., short strategies)\footnote{If $g\ge 0$, adding resources cannot reduce $\sum_{i\in S} g_i$, so minimizers tend to avoid detours; longer strategies can still win if they use sufficiently cheaper resources.}.
Let $\tau:\mathcal S\to\mathcal T$ be a user-chosen stratification statistic (e.g., $\tau(S)=|S|$) and define strata $\mathcal S_t:=\{S\in\mathcal S:\tau(S)=t\}$.
We sample in two stages: choose $t\sim w$ and then sample $S$ uniformly from $\mathcal S_t$, i.e., \(q(S)=w_{\tau(S)}/|\mathcal S_{\tau(S)}|\).
Then \(p(g)=\sum_{t\in\mathcal T} w_t\,|\mathrm{Opt}(g)\cap \mathcal S_t|/|\mathcal S_t|\), so $p(g)$ depends on within-stratum optimizer fractions rather than on $|\mathrm{Opt}(g)|/|\mathcal S|$.

\paragraph{Two length-debiased choices.}
A natural choice is $\tau(S)=|S|$, which motivates length-debiased sampling.
Let $\mathcal L:=\{|S|:S\in\mathcal S\}$ and $N_t:=|\{S\in\mathcal S:|S|=t\}|$.
We use \textbf{uniform-over-length} (UL), $q(S)=1/(|\mathcal L|\,N_{|S|})$, and \textbf{harmonic-length} (HL), $q(S)\propto 1/(|S|\,N_{|S|})$.

\subsection{Outer Loop: Goldstein Stationarity and Convergence of \Cref{alg:zo-outer}}
\label{sec:zo_convergence}

Since the equilibrium map can switch combinatorially with $\theta$, the hyper-objective $\Phi(\theta)=F(\theta,y^\star(\theta))$ is typically Lipschitz but nonsmooth.
We therefore measure progress by Goldstein stationarity and relate the outer-loop accuracy to the inner equilibrium error.

\paragraph{Stationarity measure.}
Let $\varphi:\R^k\to\R$ be locally Lipschitz and $\Theta\subseteq\R^k$ be closed and convex.
For $\rho>0$, define \(\partial_\rho \varphi(\theta):=\mathrm{conv}\big(\bigcup_{\|v-\theta\|\le\rho}\partial \varphi(v)\big)\), where $\partial \varphi(\cdot)$ is the Clarke subdifferential \citep{clarke1990optimization}; see \citet{goldstein1977optimization}.
For $\eta>0$, define \(G_\Theta(\theta,g;\eta):=\eta^{-1}\big(\theta-\Pi_\Theta(\theta-\eta g)\big)\).
We call $\theta\in\Theta$ an $(\epsilon,\rho,\eta)$-generalized Goldstein stationary point\footnote{This notion is consistent with standard first-order stationarity when $\varphi$ is smooth: if $\varphi$ is $C^1$, then the Clarke subdifferential satisfies $\partial \varphi(\theta)=\{\nabla \varphi(\theta)\}$, so $\partial_\rho \varphi(\theta)$ reduces to nearby gradients and collapses to $\{\nabla \varphi(\theta)\}$ as $\rho\downarrow 0$.} (GGSP) of $\varphi$ if \(\min_{g\in\partial_\rho \varphi(\theta)}\|G_\Theta(\theta,g;\eta)\|\le \epsilon\).

\paragraph{Inexact objective evaluations.}
The outer loop does not access $\Phi$ directly; instead it queries $\widehat\Phi_T(\theta):=F(\theta,y_T(\theta))$, where $y_T(\theta)$ is the output of running \Cref{alg:fw-inner} for $T$ steps.
The only inner-loop quantity that enters the outer-loop bound is \(\varepsilon_y:=\sup_{\theta\in\Theta}\sqrt{\E[\|y_T(\theta)-y^\star(\theta)\|^{2}]}\).
By \Cref{ass:LL_UL_regular}, $F(\theta,\cdot)$ is $L_{F,1}$-Lipschitz on $\mathcal C$, so
\(\E[|\widehat\Phi_T(\theta)-\Phi(\theta)|]\le L_{F,1}\,\E[\|y_T(\theta)-y^\star(\theta)\|]\le L_{F,1}\varepsilon_y\).
This inexactness induces a bias term in the two-point estimator used by \Cref{alg:zo-outer}.

\begin{thm}[Convergence of \Cref{alg:zo-outer} to a GGSP of $\Phi$]\label{thm:algo2_ggsp_phi_subsampled}
Let \Cref{ass:costs,ass:LL_UL_regular,ass:lsm} hold. Let $\Phi$ be $L_\Phi$-Lipschitz on $\Theta$ (\Cref{lem_solution_map_Lip}).
Fix $\rho>0$, run \Cref{alg:zo-outer} (with direction batch size $B\ge 1$) for $K$ iterations with constant stepsize
$\eta \le c_0\,\rho/(L_\Phi\sqrt{k})$, and let $\widehat{\theta}$ be a uniformly random iterate from $\{\theta_0,\dots,\theta_{K-1}\}$.
Let $\varepsilon_y$ be as defined above and define the initial hyper-objective gap $\Delta_{\Phi,0}:=\Phi(\theta_0)-\inf_{\theta\in\Theta}\Phi(\theta)$.
Then there exist universal constants $C_1,C_2,C_3>0$ such that
\begin{align}\label{eq:ggsp_bound_phi_concise_subsampled_rewrite}
\MoveEqLeft[6]
\E\!\left[\min_{g\in\partial_\rho \Phi(\widehat{\theta})}\|G_\Theta(\widehat{\theta},g;\eta)\|\right]\le
C_1\sqrt{\frac{\Delta_{\Phi,0}+2\rho L_\Phi}{\eta K}}\nonumber\\
&\;+\;
C_2\frac{k}{\rho}\,L_{F,1}\,\varepsilon_y
\;+\;
C_3\sqrt{\frac{k}{B}}\,L_\Phi.
\end{align}
In particular, for any $\epsilon>0$, if
\(\varepsilon_y \le r\) and
\(\varepsilon_y \le \rho\epsilon\,(3C_2kL_{F,1})^{-1}\),
\(K \ge 9C_1^2\big(\Phi(\theta_0)-\inf_{\Theta}\Phi+2\rho L_\Phi\big)\,\eta^{-1}\epsilon^{-2}\), and
\(B \ge 9C_3^2\,kL_\Phi^2\,\epsilon^{-2}\), then $\widehat{\theta}$ is an $(\epsilon,\rho,\eta)$-GGSP of $\Phi$ in expectation.
\end{thm}

\paragraph{Total oracle complexity.}
Each outer iteration evaluates $\widehat\Phi_T$ $2B$ times and thus makes $2BT$ subsampled-LMO calls; over $K$ iterations,
$\#(\text{inner gradients})=\#(\text{subsampled LMO calls})=2KBT$ and $\#(\text{sampled strategies})=2KBTm$.
Let $\bar T_0$ be the resulting uniform burn-in from Theorem~\ref{thm:fw_equilibrium_subsampled}.
Define $\Delta_\Phi:=\Delta_{\Phi,0}+2\rho L_\Phi$ and $c_y:=3C_2kL_{F,1}$.
Using the parameter choices from Theorem~\ref{thm:algo2_ggsp_phi_subsampled} gives $\eta=c_0\rho\,(L_\Phi\sqrt{k})^{-1}$, $B=\O(kL_\Phi^2\,\epsilon^{-2})$, and $K=\O(\Delta_\Phi\,(\eta\epsilon^2)^{-1})=\O(\Delta_\Phi\,L_\Phi\sqrt{k}\,\rho^{-1}\epsilon^{-2})$.
Let $\bar\varepsilon_y:=\min\{r,\rho\epsilon\,c_y^{-1}\}$ and $A:=L_{f,2}D^2\,\max\{1,D/r\}\,(\alpha\,\kappa_m)^{-1}$; by Theorem~\ref{thm:fw_equilibrium_subsampled}, it suffices to take $T=\O(A\,\bar\varepsilon_y^{-2}+\bar T_0)$.
Then we obtain $2KBT=\O\!\big(L_\Phi^3\,k^{3/2}\,\Delta_\Phi\,T\,\rho^{-1}\epsilon^{-4}\big)$; in particular, if $\rho\epsilon\le c_yr$, the dominant scaling is $\O(\rho^{-3}\epsilon^{-6})$.

% \paragraph{Total oracle complexity.}
% Each outer iteration calls the inner solver $2B$ times, so over $K$ iterations we use $2KBT$ inner gradients and subsampled LMO calls, and $2KBTm$ sampled strategies.
% Assume $\bar\Delta_0:=\sup_{\theta\in\Theta}\Delta_0(\theta)<\infty$ and let $\bar T_0$ be the corresponding uniform burn-in from Theorem~\ref{thm:fw_equilibrium_subsampled}.
% Fix a target $\bar\varepsilon_y\in(0,r]$; taking \(T \ge \left\lceil 4L_{f,2}D^2\,\max\{1,D/r\}/(\alpha\,\kappa_m\,\bar\varepsilon_y^2)-1\right\rceil + \bar T_0\) ensures $\varepsilon_y \le \bar\varepsilon_y$.
% For Theorem~\ref{thm:algo2_ggsp_phi_subsampled}, set $\bar\varepsilon_y:=\min\{r,\rho\epsilon/(3C_2kL_{F,1})\}$ and choose $\eta=c_0\rho/(L_\Phi\sqrt{k})$,
% \(K=\left\lceil 9C_1^2(\Phi(\theta_0)-\inf_\Theta\Phi+2\rho L_\Phi)/(\eta\epsilon^2)\right\rceil\), and \(B=\left\lceil 9C_3^2kL_\Phi^2/\epsilon^2\right\rceil\).
% Let $\Delta_\Phi:=\Phi(\theta_0)-\inf_\Theta\Phi+2\rho L_\Phi$ and \(T_{\mathrm{in}}:=L_{f,2}D^2\,\max\{1,D/r\}/(\alpha\,\kappa_m\,\bar\varepsilon_y^2)+\bar T_0\); then
% \(2KBT=\mathcal O\!\left(L_\Phi^3\,k^{3/2}\,\Delta_\Phi\,T_{\mathrm{in}}/(\rho\,\epsilon^4)\right)\).

\section{Experiments}\label{sec:experiments}

\IfFileExists{paper/exp.tex}{\begin{figure*}[t!]
\centering
\includegraphics[]{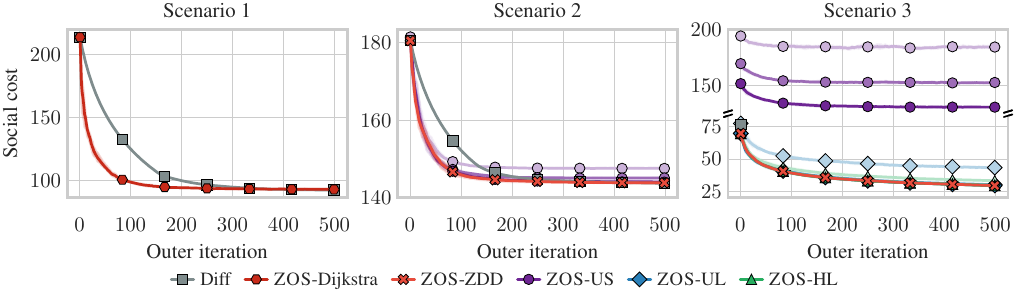}
\caption{Leader objective vs outer iterations for Scenarios~1--3. For subsampled LMOs (US/UL/HL), lighter shades denote smaller sampling budgets $m$ (we use $m\in\{10,100,1000\}$ in Scenario~2 and 3); bands are 99\% CIs over 10 runs, while Diff is deterministic.}
\label{fig:tntp-cost}
\end{figure*}

\begin{figure*}[t!]
\centering
\includegraphics[]{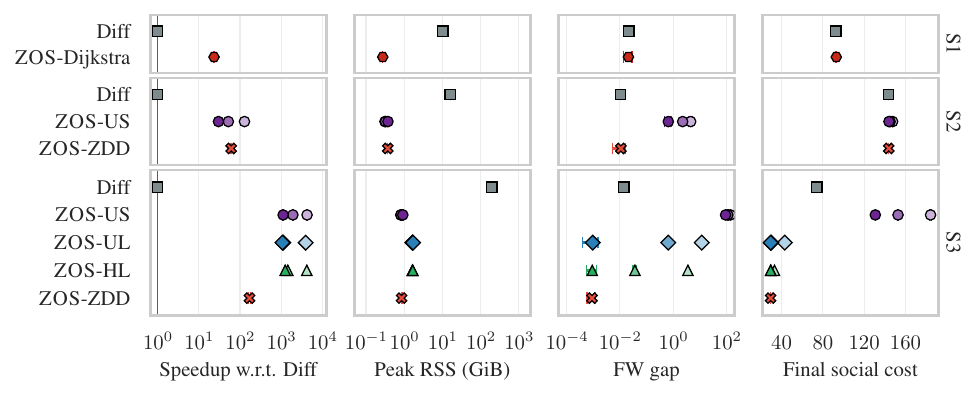}
\caption{Final-iterate diagnostics: speedup vs Diff, peak RSS, FW gap, and social cost, for Scenarios~1--3. For subsampling-based variants, lighter shades denote smaller $m$ (same $m$ as in \Cref{fig:tntp-cost}); points are means and bars are 99\% CIs over 10 runs.}
\label{fig:tntp-runtime-fw}
\end{figure*}

We evaluate \textsc{ZO-Stackelberg} on three transportation networks from the \citet{tntp} (TNTP) benchmark suite.
The scenarios isolate the role of the lower-level linear minimization oracle (LMO): from polynomial-time to NP-hard families where exact minimization is feasible only via a ZDD oracle or must be approximated by sampling.

\subsection{Experimental Setting}
\paragraph{Networks and strategy families.}
We use three TNTP networks (Winnipeg/Chicago/Philadelphia), inducing $s$--$t$ paths (Scenario~1), Hamiltonian $s$--$t$ paths (Scenario~2), and Steiner cycles through a fixed terminal set (Scenario~3).
This yields a polynomial-time LMO in Scenario~1, a fast exact ZDD oracle in Scenario~2, and a massive ZDD that motivates subsampling in Scenario~3.
\Cref{tab:tntp-settings} summarizes the resulting graphs and ZDD statistics.

\begin{table}[tb]
\centering
\caption{TNTP scenarios used in the experiments.}
\label{tab:tntp-settings}
\resizebox{\linewidth}{!}{%
\begin{tabular}{lrrrrrr}
\toprule
Scenario & Family & $|V|$ & $|E|$ & $|\mathcal S|$ & ZDD nodes & ZDD build (s) \\
\midrule
1 & $s$--$t$ paths & 49 & 82 & $3.12\times 10^{7}$ & $1.36\times 10^{4}$ & $9.12 \uerr{2.10}$ \\
2 & Hamiltonian $s$--$t$ paths & 63 & 118 & $2.53\times 10^{6}$ & $2.88\times 10^{4}$ & $4.23 \uerr{1.32}$ \\
3 & Steiner cycles & 110 & 176 & $3.81\times 10^{14}$ & $1.08\times 10^{6}$ & $301.65 \uerr{15.47}$ \\
\bottomrule
\end{tabular}%
}
\end{table}

\paragraph{Cost model and Stackelberg objective.}
In our simulation, we use the same fractional edge cost model as \citet{Sakaue2021diff}, namely $c_i(y_i;\theta_i)=d_i\bigl(1 + C\,\tfrac{y_i}{\theta_i+1}\bigr)$ with $\Theta:=\{\theta\in\R_+^n:\mathbf 1^\top\theta=n\}$ and leader objective $F(\theta,y)=\sum_i c_i(y_i;\theta_i)\,y_i$.
Each objective query solves a Wardrop equilibrium via $T{=}3000$ FW steps (see \Cref{app:exp_details} for more details).

\subsection{Algorithms}
\paragraph{Diff (\citet{Sakaue2021diff}).}
Diff replaces the LMO by a differentiable softmin over the ZDD and updates $\theta$ by backpropagating through $T$ inner iterations.

\paragraph{\textsc{ZOS} (our approach).}
We run \Cref{alg:zo-outer} with an LMO implemented via Dijkstra's algorithm (Scenario~1), an exact ZDD argmin (Scenario~2), and either an exact ZDD argmin or ZDD-guided subsampling (Scenario~3).
We use the US/UL/HL sampling schemes from \Cref{sec:inner_subsampled}.

\subsection{Results and Analysis}

\Cref{fig:tntp-cost,fig:tntp-runtime-fw} report social-cost convergence and, at the final iterate, runtime per outer iteration, peak RSS, FW gap, and social cost.
Since a small FW gap certifies approximate equilibrium, social-cost comparisons are meaningful only at small gaps.
Across all scenarios, \textsc{ZOS} attains low FW gaps and low social cost while being $\approx 20$--$1000\times$ faster per outer iteration and using at most \SI{1.7}{\GiB}, whereas Diff uses \SIrange{10}{194}{\GiB}.
The gains come from treating equilibrium computation as a black box (no backprop through $T{=}3000$ inner steps) and parallelizing independent equilibrium solves and sampled-strategy evaluations.

\paragraph{Scenario 1: Polynomial-time regime.}
Since \textsc{ZOS} treats FW as a black box, in Scenario~1 we plug in the exact shortest-path oracle and achieve Diff-level FW gaps ($\approx 2\times 10^{-2}$) with comparable social cost.
This yields a $\approx 23\times$ speedup and reduces peak RSS from \SI{10.2}{\GiB} to \SI{0.28}{\GiB}.

\paragraph{Scenario 2: NP-hard regime with tractable ZDD.}
Scenario~2 is NP-hard, so both methods rely on a ZDD; however, the ZDD is small enough that exact LMO minimization remains inexpensive.
Exact \textsc{ZOS}--ZDD matches Diff's FW gap (about $10^{-2}$) and achieves comparable social cost, while being roughly $61\times$ faster and using far less memory (\SI{0.37}{\GiB} vs.\ \SI{15.9}{\GiB} peak RSS).
This gap is largely explained by Diff backpropagating through the ZDD across all $T$ inner iterations, whereas \textsc{ZOS} treats the equilibrium solver as a black box and never differentiates through it.
The subsampling variants\footnote{Since all Hamiltonian $s$--$t$ paths have the same length, the three sampling schemes coincide in this scenario.} attain similarly low social cost, but their FW gaps plateau around $10^0$ and they offer little speedup over exact \textsc{ZOS}--ZDD.
Overall, when the ZDD is tractable (e.g., fewer than $10^5$ nodes), exact \textsc{ZOS}--ZDD is the recommended choice.

\paragraph{Scenario 3: NP-hard regime with massive ZDD.}
This setting yields a massive ZDD, making repeated exact minimization costly; subsampling is therefore essential.
Uniform-strategy sampling (US) essentially stalls in both social cost and FW gap because the probability of hitting an LMO optimizer is vanishingly small (indeed, $\kappa_m \approx m/|\mathcal S|$).
In contrast, our stratified schemes (UL/HL) exploit the structure induced by nonnegative weights: when $g\ge 0$, adding resources cannot decrease the linear cost, so LMO minimizers are typically concentrated in smaller-cardinality (short) strata.
By over-sampling these strata, UL/HL dramatically increase the effective optimizer hit-rate and, once $m$ is large enough, match exact \textsc{ZOS}--ZDD in social cost and FW gap (FW gap $\approx 10^{-3}$ at $m{=}1000$).
Exact \textsc{ZOS}--ZDD attains similar accuracy but remains expensive (about $37$s per outer iteration), whereas UL/HL at $m{=}1000$ are roughly $6$--$7\times$ faster, with $<\SI{1}{\GiB}$ additional memory.
Diff is impractical in this regime (\SI{194}{\GiB} peak) and failed to complete more than five outer iterations within 10 hours.

\typeout{MARK:END_EXPERIMENTS page=\thepage}
}{}

\bibliography{bcgg_bib}
\bibliographystyle{plainnat}

\newpage
\appendix
\crefname{appendix}{appendix}{appendices}
\Crefname{appendix}{Appendix}{Appendices}
\crefalias{section}{appendix}
\crefalias{subsection}{appendix}
\crefalias{subsubsection}{appendix}

\section{Notation}\label{app:notation}
The table below summarizes the main notation used throughout the paper.

\begin{center}
\begin{adjustbox}{max width=\textwidth}
\small
\begin{tabular}{@{}p{0.22\textwidth}p{0.74\textwidth}@{}}
\toprule
\textbf{Symbol} & \textbf{Meaning} \\
\midrule
$n$; $[n]$ & Number of resources; index set $\{1,\dots,n\}$. \\
$\mathcal S$; $d:=|\mathcal S|$ & Set of feasible (combinatorial) strategies; number of strategies. \\
$S\in\mathcal S$; $\mathbf 1_S\in\{0,1\}^n$ & A strategy (subset of resources); its incidence vector. \\
$\Delta^d$; $z\in\Delta^d$ & Probability simplex; population distribution over strategies. \\
$y(z)\in\R^n$; $\mathcal C$ & Induced load $y(z)=\sum_{S} z_S\mathbf 1_S$; feasible load polytope $\mathcal C=\mathrm{conv}\{\mathbf 1_S:S\in\mathcal S\}\subseteq[0,1]^n$. \\
$c_i(\cdot;\theta)$; $c_S(y)$ & Cost of resource $i$ under parameter $\theta$; strategy cost $c_S(y)=\sum_{i\in S} c_i(y_i;\theta)$. \\
$\theta\in\Theta\subseteq\R^k$; $k$ & Leader parameter and its feasible set; dimension. \\
$f(\theta,y)$ & Lower-level (Beckmann) potential; Wardrop equilibria are minimizers over $\mathcal C$. \\
$y^\star(\theta)$ & Wardrop equilibrium load: $y^\star(\theta)\in\argmin_{y\in\mathcal C} f(\theta,y)$. \\
$F(\theta,y)$; $\Phi(\theta)$ & Upper-level objective; hyper-objective $\Phi(\theta)=F(\theta,y^\star(\theta))$. \\
$y_t$; $T$ & Inner Frank--Wolfe iterates; number of inner iterations. \\
$g_t=\nabla_y f(\theta,y_t)$ & FW gradient used in the linear minimization oracle (LMO). \\
$\mathrm{LMO}(g)$; $\mathrm{LMO}_m(g)$ & Exact LMO $\argmin_{y\in\mathcal C}\langle g,y\rangle$; subsampled LMO over $m$ sampled strategies. \\
$q$; $S^{(1:m)}\sim q$; $m$ & Sampling distribution over $\mathcal S$; sampled candidates; sample size. \\
$\mathrm{Opt}(g)$; $p(g)$; $\kappa_m(g)$ & Exact LMO minimizers; optimizer mass $p(g)=q(\mathrm{Opt}(g))$; optimizer-hit probability $\kappa_m(g)=1-(1-p(g))^m$. \\
$\underline p$; $\kappa_m$ & Uniform lower bound $p(g)\ge\underline p$ along the inner run; corresponding uniform hit probability $\kappa_m=1-(1-\underline p)^m$. \\
$\tau:\mathcal S\to\mathcal T$; $\mathcal S_t$; $w$ & Stratification map; stratum $\mathcal S_t:=\{S:\tau(S)=t\}$; stratum weights $w=(w_t)_{t\in\mathcal T}$. \\
$\mathcal L$; $N_t$ & Set of attainable lengths $\mathcal L=\{|S|:S\in\mathcal S\}$; stratum size $N_t:=|\mathcal S_t|$. \\
$\Delta_0$ & Initial inner optimality gap: $\Delta_0=f(\theta,y_0)-f(\theta,y^\star(\theta))$. \\
$\bar\Delta_0$; $\bar T_0$ & Uniform bound $\bar\Delta_0=\sup_{\theta\in\Theta}\Delta_0(\theta)$; corresponding uniform burn-in bound $\bar T_0$ in Theorem~\ref{thm:fw_equilibrium_subsampled}. \\
$T_0$ & Logarithmic burn-in length in Theorem~\ref{thm:fw_equilibrium_subsampled}. \\
$D$ & Diameter of $\mathcal C$: $D=\max\{\|u-v\|:u,v\in\mathcal C\}$. \\
$\|\cdot\|$; $\langle\cdot,\cdot\rangle$; $|\cdot|$ & Euclidean norm ($\|\cdot\|_2$) and inner product; absolute value for scalars. \\
$\theta_t$; $K$ & Outer (leader) iterates; number of outer iterations. \\
$\Pi_\Theta(\cdot)$; $\Theta_\rho$ & Euclidean projection onto $\Theta$; $\Theta_\rho:=\{\theta\in\Theta:\theta+\rho\mathbb B\subseteq\Theta\}$. \\
$\widehat\Phi_T(\theta)$ & Approximate hyper-objective value using $T$ inner steps: $\widehat\Phi_T(\theta)=F(\theta,y_T(\theta))$. \\
$\varepsilon_y$ & Uniform mean-square inner accuracy: $\varepsilon_y:=\sup_{\theta\in\Theta}\sqrt{\E[\|y_T(\theta)-y^\star(\theta)\|^{2}]}$. \\
$\widehat g_t$; $B$ & Two-point gradient estimator in Algorithm~\ref{alg:zo-outer}; direction mini-batch size. \\
$\rho$; $\mathbb B$; $\mathbb S^{k-1}$ & Smoothing radius; unit ball; unit sphere in $\R^k$. \\
$u\sim\mathrm{Unif}(\mathbb S^{k-1})$ & Random direction used in two-point finite differences. \\
$\eta$ & Outer stepsize. \\
$\partial_\rho \Phi(\theta)$; $G_\Theta(\theta,g;\eta)$ & Goldstein $\rho$-subdifferential; projected gradient mapping. \\
$L_{f,2},L_{F,1},L_\Phi$; $\alpha,r$ & Smoothness/Lipschitz constants; quadratic-growth parameters. \\
\bottomrule
\end{tabular}
\end{adjustbox}
\end{center}
\section{Related Work}\label{app:related_work}

\paragraph{Stackelberg control and congestion games.}
Congestion games and selfish routing are classical models for large-scale transportation and communication systems, with Wardrop equilibria as the standard notion of flow-level equilibrium \citep{wardrop1952traffic}.
The potential formulation of Wardrop equilibria and its connection to convex optimization dates back to the Beckmann transformation \citep{beckmann1956studies} and has since been widely used in algorithmic and economic analyses of routing \citep{roughgarden_tardos_2002_selfish_routing,sandholm_population_2010}.
Stackelberg control in these models includes leader decisions such as tolling, pricing, or capacity design that influence equilibrium behavior; see, e.g., \citet{karakostas_kolliopoulos_2009_stackelberg} and references therein.
More general Stackelberg formulations in congestion games allow richer action sets (e.g., multiple user classes and non-singleton actions), further motivating combinatorial strategy spaces beyond simple path choices \citep{marchesi2019leadership}.
Related work also studies how monetary interventions can steer the system between equilibria \citep{huang2023equilibrium_transition}.
For large networks, equilibrium computation is typically performed using methods that exploit shortest-path computations and other combinatorial subroutines.
Classic traffic-assignment algorithms include Frank--Wolfe-type schemes specialized to the Beckmann potential \citep{leblanc1975efficient}.
Since Wardrop equilibria coincide with minimizers of the (convex) Beckmann potential, each Frank--Wolfe step solves the linearization of the potential over the feasible flow polytope; in network models this reduces to an all-or-nothing assignment subproblem (typically implemented via shortest-path computations).
This yields a projection-free method whose per-iteration cost is dominated by combinatorial primitives rather than Euclidean projections, which is one reason Frank--Wolfe is widely used in practice.
Comprehensive discussions of step-size rules (e.g., exact line search vs.\ method-of-successive-averages heuristics), convergence, and practical variants can be found in standard references such as \citet{sheffi1984urban,patriksson1994traffic}.

\paragraph{Compact representations of combinatorial strategy sets.}
When the strategy family $\mathcal S$ is exponentially large, one approach is to represent $\mathcal S$ implicitly via compact data structures that share substructure across many feasible strategies.
Binary decision diagrams (BDDs) encode a Boolean function as a reduced, ordered decision DAG \citep{bryant1986graph}; in our setting, the function is simply the indicator of feasibility over subsets of $[n]$.
Zero-suppressed BDDs (ZDDs) modify the reduction rules to better match \emph{families of sparse sets}, roughly, instances where most sets omit most items, which often leads to dramatically smaller diagrams in practice \citep{minato_1993_zdd}.
Once constructed, the ZDD supports efficient dynamic programming on the DAG to perform common set-family operations (e.g., union/intersection), to count feasible strategies, to sample feasible strategies, and to optimize additive costs over $\mathcal S$ without enumerating it \citep{minato_1993_zdd,bergman2018decision}.
Although ZDD size is problem- and ordering-dependent (and can be exponential in the worst case), decision-diagram techniques have become a practical tool for discrete optimization and for implementing combinatorial oracles over large strategy spaces \citep{bergman2018decision}. 
Conceptually, while our use of ZDDs is similar to the equilibrium computations in \citet{Sakaue2021diff}, our role for them is different. We use ZDDs purely as a \emph{combinatorial oracle} for linear minimization inside Frank--Wolfe; the inner loop remains the standard (non-smoothed) FW method, and we never require its iterations to be differentiable with respect to~$\theta$. This stands in contrast to differentiable approaches that modify the FW steps (e.g., via softmin relaxations) to enable backpropagation through the equilibrium map and rely on gradient-based updates for the leader's problem. Our algorithm instead keeps the equilibrium computation as a black box and handles the non-smooth dependence on~$\theta$ entirely at the outer level via zeroth-order methods.

\paragraph{Differentiating through equilibria and decision diagrams.}
A contrasting line of work makes the equilibrium computation itself \emph{differentiable} and then applies gradient-based bilevel optimization by differentiating through the lower-level computation.
For Stackelberg CCGs, \citet{Sakaue2021diff} leverage ZDDs to represent $\mathcal S$ and develop a differentiable Frank--Wolfe-style equilibrium solver via a softmin relaxation of the LMO, enabling automatic differentiation of the surrogate map $\theta\mapsto y_T(\theta)$ and the corresponding surrogate objective value $F(\theta,y_T(\theta))$.
More broadly, differentiable optimization layers enable backpropagation through optimization problems \citep{amos2017optnet,agrawal2019differentiable}.
This yields a practical hyper-gradient pipeline and strong empirical performance, but the optimized objective is a smoothed surrogate rather than the true hyper-objective $\Phi(\theta)=F(\theta,y^\star(\theta))$.
Establishing end-to-end guarantees for $\Phi$ requires quantifying both the equilibrium approximation error $y_T(\theta)\approx y^\star(\theta)$ and the mismatch between the gradient of surrogate ($\nabla F(\theta, y_T(\theta))$) and suitable (generalized) derivatives of $\Phi$; this is explicitly noted as challenging in \citet{Sakaue2021diff}.
In contrast, our method avoids differentiating through equilibria and provides convergence to generalized Goldstein stationary points of the (typically Lipschitz but nonsmooth) hyper-objective, with an explicit dependence on the inner equilibrium error.

\paragraph{Zeroth-order (derivative-free) optimization.}
Because the hyper-objective $\Phi(\theta)=F(\theta,y^\star(\theta))$ is typically \emph{Lipschitz but nonsmooth} due to active-set changes in $y^\star(\theta)$, we adopt a weak first-order stationarity notion based on Clarke and Goldstein subdifferentials \citep{clarke1990optimization,goldstein1977optimization}.
These generalized derivatives and stationarity measures also underpin classic algorithms for nonsmooth nonconvex optimization, most notably gradient sampling \citep{burke2005robust}, as well as modern analyses of stochastic methods on tame/nonsmooth landscapes \citep{davis2019stochastic,liu2024zo_constrained_nonsmooth}.
Our outer loop is a zeroth-order method based on \emph{randomized smoothing} and \emph{two-point} finite-difference estimators, a standard approach in derivative-free optimization and bandit/zeroth-order learning: foundational results include \citet{flaxman2005online}, with refined rates and smoothing-based analyses developed in, e.g., \citet{duchi2015optimal,nesterov2017random,ghadimi2013stochastic}.
Zeroth-order methods have also been analyzed in game-theoretic settings, including convex--concave minmax formulations and (strongly) monotone games \citep{maheshwari2021zo_minmax,drusvyatskiy2022df_gradient_play}.
More classical derivative-free schemes include stochastic approximation methods such as SPSA \citep{spall1992multivariate} and general-purpose trust-region frameworks for derivative-free optimization \citep{conn2009introduction}.
Zeroth-order approaches are also useful in bilevel problems when the outer objective is evaluated through an inner solve and may be nonsmooth or otherwise unsuitable for backpropagation; see, e.g., \citet{chen2025set,masiha2025superquantile}.

\paragraph{Randomized and approximate Frank--Wolfe oracles.}
To reduce the cost of LMO calls when $\mathcal S$ is very large, several works study conditional-gradient methods with randomized or approximate oracles.
In particular, \citet{kerdreux2018fw_subsampling} analyze a subsampling-oracle Frank--Wolfe method under a uniform-inclusion model, obtaining an $\O(1/(\eta T))$ rate.
Related variants of stochastic/online Frank--Wolfe have been studied in different settings, including projection-free stochastic optimization \citep{hazan2016variance,mokhtari2018stochasticfw}.
Our analysis isolates a single optimizer-hit quantity $\kappa_m$ and allows non-uniform, structure-aware sampling distributions tailored to imbalanced strategy spaces.

\section{Applications of Stackelberg Control in Congestion Games}\label{app:applications}

Stackelberg control models settings in which a system operator (leader) chooses a parameter vector $\theta$ that affects congestion-dependent costs.
A large population of users then responds by selecting feasible combinatorial strategies, resulting in a Wardrop equilibrium.
To connect these applications to the bilevel model in \Cref{sec:2.3}, we make the correspondence to our primitives explicit.
\begin{itemize}[leftmargin=0pt, topsep=0pt, partopsep=0pt, parsep=0pt, itemsep=1pt]
\item \textbf{Leader parameter $\theta$.} $\theta\in\Theta\subseteq\R^k$ collects the operator's controls (e.g., link tolls, capacity expansions, or per-link prices) that enter follower costs.
\item \textbf{Equilibrium load $y^\star(\theta)$.} $y\in\mathcal C\subseteq[0,1]^n$ is the resource-load vector induced by a population distribution over strategies, and $y_i$ is the fraction of total demand using resource $i$ (after normalizing demand so the total mass is $1$).
\item \textbf{Resource costs $c_i(\cdot;\theta)$.} The perceived cost of using resource $i$ at load $y_i$ is $c_i(y_i;\theta)$, assumed continuous and strictly increasing in $y_i$ (\Cref{ass:costs}) and continuously differentiable in $\theta$.
\item \textbf{Potential $f(\theta,y)$.} Wardrop equilibria coincide with minimizers of the Beckmann potential $f(\theta,y)=\sum_i\int_0^{y_i} c_i(u;\theta)\,du$ over $\mathcal C$.
\item \textbf{Leader objective $F(\theta,y)$.} $F$ evaluates the equilibrium outcome (e.g., total delay, emissions proxies, or revenue and investment costs) and defines the hyper-objective $\Phi(\theta)=F(\theta,y^\star(\theta))$.
\end{itemize}
Below we give representative real-world domains where this template is standard.
In all of them, $\Theta$ naturally encodes operational, regulatory, or budget constraints on the leader decisions.
In the models below, the congestion sensitivity $\partial_y c_i(y_i;\theta)$ admits a uniform positive lower bound over $\Theta\times[0,1]$.
In each case, $\theta$ enters costs additively (as tolls, prices, or fees), so the \emph{congestion slope}, the derivative of the cost with respect to load, $\partial_y c_i(y_i;\theta)$, does not depend on $\theta$; hence this uniform lower bound can be chosen independent of $\theta$.
As a result, the Beckmann potential is uniformly strongly convex in $y$.
This implies the quadratic-growth property around $y^\star(\theta)$ used in \Cref{sec:2.3}.
The same smooth parametric dependence makes $(\theta,y)\mapsto \nabla_y f(\theta,y)$ Lipschitz on $\Theta\times\mathcal C$, and the objectives $F$ written below are Lipschitz on the same domain.

\paragraph{Urban traffic pricing (link tolls and cordon charges).}
Resources are road segments or lanes, and strategies $S\in\mathcal S$ are feasible routes (e.g., $s$--$t$ paths) that drivers may take.
For a fixed origin--destination pair $(s,t)$, this is the tractable family $\mathcal S=\mathcal P_{s\to t}$ of feasible $s$--$t$ paths.
The leader parameter $\theta=(\tau_i)_{i\in[n]}$ can represent nonnegative tolls or charges applied to each road segment.
The induced equilibrium load $y^\star(\theta)$ is the vector of equilibrium link flows, where $y_i$ is the fraction of total demand traversing link $i$.
A common separable perceived cost is
\[
  c_i(y_i;\theta)=t_i^0 + a_i\,y_i + \tau_i,
\]
with free-flow time $t_i^0>0$ and congestion sensitivity $a_i>0$.
The lower-level potential is the Beckmann objective
\[
  f(\theta,y)=\sum_{i\in[n]}\int_0^{y_i} c_i(u;\theta)\,du
  =\sum_{i\in[n]}\Bigl(t_i^0 y_i+\tfrac{a_i}{2}y_i^2+\tau_i y_i\Bigr),
\]
and its minimizer over $\mathcal C$ is precisely the user-equilibrium load $y^\star(\theta)$.
The leader objective $F(\theta,y)$ is a system-level goal such as minimizing total travel time (excluding toll transfers) subject to constraints on toll magnitudes or revenue.
For example, one can take $F(\theta,y)=\sum_i y_i(t_i^0+a_i y_i)+\lambda\|\theta\|_2$ for a regularization parameter $\lambda>0$.
Equilibrium-based road pricing formulations explicitly couple $\theta$ and $y^\star(\theta)$ \citep{yang_bell_1997_traffic_restraint,dial1999minimal_revenue,hearn2001computational_toll_pricing}, and real-world deployments such as the Stockholm congestion charge highlight the practical relevance \citep{eliasson2009stockholm_overview}.

\paragraph{Telecommunication networks (multicast routing and per-link pricing).}
Resources are links or routers in a communication network.
In content distribution and streaming, each session must connect a source to a set of receivers.
A natural combinatorial strategy is a feasible multicast routing tree.
Thus $\mathcal S$ can be taken as the family of feasible multicast trees (Steiner trees) connecting the terminals \citep{kim_choo_mutka_lim_park_2013_qos_multicast_steiner_trees}.
Here $y_i$ is the fraction of sessions whose chosen multicast tree uses link $i$.
 The leader parameter $\theta$ can encode per-link prices or fees that influence how traffic is routed \citep{yaiche2000bandwidth_pricing,korilis1997achieving}.
A common separable per-link cost model is motivated by queueing delay, which increases sharply as utilization approaches capacity.
After normalizing load by a (fixed) design capacity with a safety margin, a simple form is
\[
  c_i(y_i;\theta)=\ell_i^0 + \frac{b_i}{\delta_i+1-y_i} + p_i,
\]
with baseline latency $\ell_i^0>0$, congestion parameter $b_i>0$, slack $\delta_i>0$, and a price component $p_i$ included in $\theta$.
The associated Beckmann potential is
\[
  f(\theta,y)
  =\sum_{i\in[n]}\int_0^{y_i} c_i(u;\theta)\,du
  =\sum_{i\in[n]}\Bigl(\ell_i^0 y_i + b_i\log\Big(\frac{\delta_i+1}{\delta_i+1-y_i}\Big)+p_i y_i\Bigr),
\]
and its minimizer over $\mathcal C$ is the Wardrop equilibrium load $y^\star(\theta)$.
Representative upper-level objectives trade off aggregate latency, operational costs, and revenue, e.g.,
\(
F(\theta,y)=\sum_i y_i\bigl(\ell_i^0+\frac{b_i}{\delta_i+1-y_i}\bigr)-\gamma\sum_i p_i y_i
\)
for a revenue weight $\gamma>0$.
Game-theoretic models of bandwidth allocation and pricing commonly take this equilibrium-response form \citep{yaiche2000bandwidth_pricing}.

\paragraph{Freight logistics (multi-stop delivery tours and Steiner cycles).}
In distribution and pickup--delivery logistics, a carrier often runs a tour that starts at a depot, visits a prescribed set of stops, and returns to the depot.
This is the basic modeling primitive behind vehicle routing \citep{dantzig_ramser_1959_truck_dispatching}.
When tours are chosen over a road network, a natural combinatorial strategy is a cycle subgraph that contains all required stops.
Equivalently, $\mathcal S$ can be taken as the family of Steiner cycles, i.e., simple cycles that contain a fixed set of terminal locations but may traverse additional intersections.
Steiner traveling-salesman and Steiner-cycle formulations formalize this structure \citep{interian_ribeiro_2017_steiner_tsp,salazar_gonzalez_2003_steiner_cycle_polytope}.
Here $y_i$ is the fraction of tours that traverse road segment $i$.
The leader can influence these tours via per-segment access fees for freight vehicles, $\theta=(\tau_i)_{i\in[n]}$.
One cost model that captures both congestion and the higher externalities of freight traffic is
\[
  c_i(y_i;\theta)=t_i^0 + a_i\,y_i + b_i\,y_i^2 + \tau_i,
\]
with $t_i^0>0$, $a_i>0$, and $b_i\ge 0$.
The lower-level potential remains $f(\theta,y)=\sum_i\int_0^{y_i} c_i(u;\theta)\,du$.
A representative leader objective is $F(\theta,y)=\sum_i y_i\,c_i(y_i;\theta)+\lambda\|\theta\|_2$.

Across these domains, the strategy family $\mathcal S$ (all feasible routes, schedules, or routing trees) can be exponentially large, so the equilibrium computation and the resulting hyper-objective optimization must rely on combinatorial oracles rather than explicit enumeration.
In the road-pricing example with $\mathcal S=\mathcal P_{s\to t}$, the linear minimization step in Frank--Wolfe is a shortest-path problem.
In the multicast example, the corresponding oracle is a minimum-cost multicast tree, i.e., a Steiner-tree problem.
In the freight-logistics example, it is a minimum-cost Steiner cycle.
These NP-hard families motivate oracle implementations based on compact representations such as zero-suppressed decision diagrams (ZDDs).
Moreover, as the leader perturbs $\theta$, the identity of minimum-cost strategies at equilibrium can change discretely, inducing kinks in the hyper-objective $\Phi(\theta)$ even when $c_i$ and $F$ are smooth.
These are exactly the two challenges addressed by our method: an oracle-based FW inner loop for equilibrium computation and a zeroth-order outer loop for nonsmooth bilevel optimization.

\section{Experimental Details}\label{app:exp_details}

\paragraph{Hardware and software.}
All experiments were run on an HPC cluster node with an AMD EPYC 9334 CPU (32 cores / 64 threads), running Red Hat Enterprise Linux~9.4, and equipped with 371\,GB of RAM.
All methods were implemented in Python, with source available at \githubrepo.
We use NetworkX for graph preprocessing and shortest-path computations.
ZDD construction uses the Graphillion library,\footnote{\url{https://github.com/takemaru/graphillion}} but we implemented our own shortest-path and sampling routines on ZDDs to support weighted sampling and to avoid unnecessary dependencies.
We implemented the differentiable baseline of \citet{Sakaue2021diff} ourselves in Python (PyTorch), since the released implementation is not in Python and was not straightforward to integrate into our pipeline.

\paragraph{TNTP networks and scenario construction.}
We start from three networks in the TNTP benchmark suite~\citep{tntp} (Winnipeg, Chicago-Sketch, and Philadelphia).
TNTP provides directed networks; since Graphillion represents undirected edge sets, we form an undirected graph by merging antiparallel arcs.
We work on the largest connected component and then select a connected subgraph (and endpoints/terminals) by an offline randomized search targeting reasonable ZDD sizes; the resulting instances are fixed across all runs.
\Cref{fig:tntp-networks} visualizes the three subgraphs used in Scenarios~1--3.

\begin{figure}[htb]
\centering
\begin{subfigure}[t]{0.48\linewidth}
\centering
\includegraphics[width=\linewidth,height=0.25\textheight,keepaspectratio]{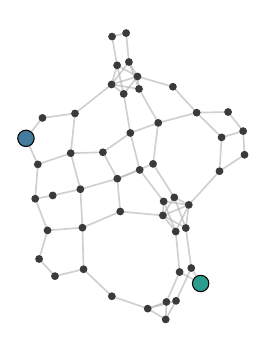}
\caption{Scenario 1 (Winnipeg): $s$--$t$ paths (source $s$ in green; target $t$ in blue).}
\end{subfigure}
\hfill
\begin{subfigure}[t]{0.48\linewidth}
\centering
\includegraphics[width=\linewidth,height=0.25\textheight,keepaspectratio]{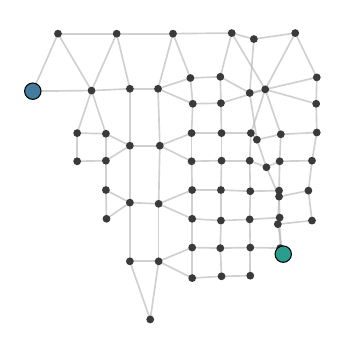}
\caption{Scenario 2 (Chicago-Sketch): Hamiltonian $s$--$t$ paths (source $s$ in green; target $t$ in blue).}
\end{subfigure}

\vspace{0.5em}
\begin{subfigure}[t]{0.98\linewidth}
\centering
\includegraphics[width=0.70\linewidth,height=0.30\textheight,keepaspectratio]{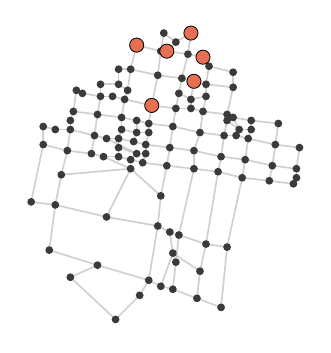}
\caption{Scenario 3 (Philadelphia): Steiner cycles (terminals highlighted).}
\end{subfigure}

\caption{TNTP-derived subgraphs used in Scenarios~1--3.}
\label{fig:tntp-networks}
\end{figure}

\paragraph{Cost model and parameters.}
All experiments use the fractional cost family
\[
  c_i(y_i;\theta_i)=d_i\Bigl(1+C\,\frac{y_i}{\theta_i+1}\Bigr),
\]
with leader feasible set $\Theta=\{\theta\in\R_+^n:\mathbf 1^\top\theta=n\}$.
For Scenarios~1--2, $d_i$ is the normalized free-flow travel time extracted from the TNTP network files (after symmetrization), scaled so that $\max_i d_i = 1$.
For Scenario~3, $d_i$ is based on Euclidean edge lengths computed from the provided node coordinates.
The scaling parameter is $C=500$ (Scenario~1), $C=20$ (Scenario~2), and $C=10$ (Scenario~3).

\paragraph{Frank--Wolfe gap.}
For a fixed leader parameter $\theta$ and inner iterate $y_t$, we report the (exact) Frank--Wolfe duality gap
\[
  g_t := \langle \nabla_y f(\theta,y_t),\, y_t - s_t^\star\rangle,
  \qquad
  s_t^\star \in \arg\min_{s\in\mathcal C}\langle \nabla_y f(\theta,y_t),\, s\rangle.
\]
Since $\nabla_y f(\theta,y_t)=c(y_t;\theta)$ in our congestion model, $g_t$ is computed by a single exact oracle call (shortest-path for Scenario~1; exact ZDD minimization for Scenarios~2--3).
The reported ``final FW gap'' corresponds to $g_T$ at the last inner iteration at the final outer iterate.

\paragraph{Algorithm details.}
For \textsc{ZOS}, we use a two-point estimator with direction batch size $B=4$, smoothing radius $\rho=0.05$, and constant outer stepsize $\eta=0.05$.
In our implementation, the random directions are sampled as i.i.d.\ coordinatewise Rademacher vectors.
We use $4$ threads to parallelize the $B=4$ function evaluations per outer iteration, and $16$ threads parallelize sampling $m$ strategies inside the inner loop.
For the differentiable baseline (Diff), we use outer gradient descent with learning rate $0.1$, and an accelerated differentiable Frank--Wolfe inner routine with stepsize $10^{-3}$.
These hyperparameters were selected via preliminary tuning to yield good empirical performance for each method.
Although Diff is implemented in PyTorch, in our setup running it on GPU did not yield a wall-clock speedup (and was often slower).
We attribute this to the irregular, memory-bound dynamic programming passes over the ZDD (bottom-up and top-down), together with kernel-launch and host--device overheads.

\paragraph{Inner-loop stepsizes and reproducibility.}
All Wardrop equilibria are computed using $T{=}3000$ steps of Frank--Wolfe.
All methods are run for up to $500$ outer iterations, with a wall-clock time limit of $10$ hours per run; in Scenario~3, Diff only completes $5$ outer iterations within this limit.
For methods with algorithmic randomness, we report averages over $10$ runs with different random seeds; Diff is deterministic.
For scenarios requiring ZDDs (Scenarios~2--3, and the Diff baseline in Scenario~1), we construct each ZDD once and cache it for reuse across runs.

\section{Proofs of \Cref{sec:2}}

\subsection{Active-Set Changes Lead to Kinks in the Equilibrium Map}\label{append:many_kinks}
In our Stackelberg setting, the leader optimizes a hyper-objective that depends on the
equilibrium response $\theta\mapsto y^\star(\theta)$. Even with smooth costs and a unique
Wardrop equilibrium, this response map can be only \emph{piecewise smooth}: kinks occur
when the set of strategies used at equilibrium changes. The example below shows that the
number of such nondifferentiable points can scale linearly with the number of strategies.
This is beneficial for our setting because it shows that, as the number of strategies grows, the hyper-objective may develop many nondifferentiable points through equilibrium switching, motivating zeroth-order methods
that do not rely on differentiating through equilibria.

\begin{example}[Many kinks scaling with the number of strategies]\label{append_ex:many_kinks}
Fix $n\ge 3$ and consider a parallel network with $n$ resources and unit demand.
The strategy set is $\mathcal S=\{\{i\}: i\in[n]\}$ (so $|\mathcal S|=n$) and the feasible
load set is
\[
\mathcal C=\Bigl\{y\in\R^n_{\ge 0}:\ \sum_{i=1}^n y_i=1\Bigr\}.
\]
Let $M>2$ and define smooth costs
\[
c_i(y_i;\theta)=y_i + M\,\phi_i(\theta),
\qquad
\phi_i(\theta):=-i\theta+\frac{i(i-1)}{2},
\qquad i\in[n],
\]
where $\theta\in\R$ is a scalar parameter. The Beckmann potential is
\[
f(\theta,y)=\sum_{i=1}^n\Bigl(\tfrac12 y_i^2 + M\,\phi_i(\theta)\,y_i\Bigr),
\qquad y\in\mathcal C.
\]
For each $\theta$, the Wardrop equilibrium $y^\star(\theta)$ is the unique minimizer of
$\min_{y\in\mathcal C} f(\theta,y)$ (uniqueness holds since $f(\theta,\cdot)$ is $1$-strongly convex).

KKT conditions (necessary and sufficient here) imply that there exists a scalar $\tau(\theta)$ such that
\[
y_i^\star(\theta)=\bigl[\tau(\theta)-M\phi_i(\theta)\bigr]_+,
\qquad
\sum_{i=1}^n \bigl[\tau(\theta)-M\phi_i(\theta)\bigr]_+=1,
\]
where $[t]_+=\max\{t,0\}$. Hence $y^\star(\theta)$ is piecewise affine in $\theta$, and kinks occur when some
coordinate hits $0$ (i.e., the active set changes).

\paragraph{Which $\phi_i$ is smallest.}
For $\theta\in[i-1,i]$, the function $\phi_i(\theta)$ is smaller than both neighbors:
\[
\phi_{i-1}(\theta)-\phi_i(\theta)=\theta-(i-1)\ge 0,
\qquad
\phi_{i+1}(\theta)-\phi_i(\theta)=i-\theta\ge 0.
\]
By chaining these inequalities, we get $\phi_i(\theta)\le \phi_j(\theta)$ for \emph{all} $j\in[n]$ whenever $\theta\in[i-1,i]$.
So as $\theta$ increases from $0$ to $n-1$, the smallest $\phi_{i}(\theta)$ moves from strategy $1$ to $2$ to $\cdots$ to $n$.

\paragraph{Equilibrium uses only strategies within $1/M$ of the minimum.}
If two strategies $i$ and $j$ are both used (i.e., $y_i^\star(\theta),y_j^\star(\theta)>0$), Wardrop conditions give equal costs:
\[
y_i^\star(\theta)+M\phi_i(\theta)=y_j^\star(\theta)+M\phi_j(\theta)
\quad\Longrightarrow\quad
y_i^\star(\theta)-y_j^\star(\theta)=M\bigl(\phi_j(\theta)-\phi_i(\theta)\bigr).
\]
Since $y_i^\star,y_j^\star\in[0,1]$, the left-hand side lies in $[-1,1]$, hence any used $j$ must satisfy
\[
\bigl|\phi_j(\theta)-\phi_i(\theta)\bigr|\le \frac1M.
\]
In other words, only strategies whose offsets are within $1/M$ of the minimum can be active.

\paragraph{An explicit ``switching window'' around each integer.}
Fix $i\in\{1,\dots,n-2\}$. Around $\theta=i$, the two closest offsets are $\phi_i$ and $\phi_{i+1}$, and
\[
\phi_{i+1}(\theta)-\phi_i(\theta)=i-\theta.
\]
For $\theta\in[i-1/M,\,i+1/M]$, this difference has magnitude at most $1/M$, so only strategies $i$ and $i+1$
can be active. Solving Wardrop with $y_{i+1}=1-y_i$ gives
\[
y_i^\star(\theta)=\frac12+\frac{M}{2}(i-\theta),
\qquad
y_{i+1}^\star(\theta)=\frac12-\frac{M}{2}(i-\theta),
\qquad
\theta\in\bigl[i-\tfrac1M,\ i+\tfrac1M\bigr],
\]
and all other coordinates are $0$.
Outside this window, the equilibrium is pure:
\[
y^\star(\theta)=e_i \ \text{for }\theta\in[i-1+\tfrac1M,\ i-\tfrac1M],
\qquad
y^\star(\theta)=e_{i+1} \ \text{for }\theta\in[i+\tfrac1M,\ i+1-\tfrac1M],
\]
where $e_i$ is the $i$th standard basis vector in $\R^n$.

Thus $y^\star(\theta)$ is \emph{continuous} and piecewise linear, but it has kinks at the two endpoints
$\theta=i\pm \tfrac1M$ for each $i=1,\dots,n-2$. Therefore, the equilibrium map has at least
$2(n-2)=\Omega(|\mathcal S|)$ nondifferentiable points on $[0,n-1]$.

Even with smooth costs and a unique Wardrop equilibrium, the equilibrium response map can have a number of kinks
that grows linearly with the number of strategies. Since our hyper-objective composes a smooth loss with $y^\star(\theta)$,
it inherits this nonsmoothness, making global differentiability an unrealistic assumption at scale and motivating
zeroth-order methods that do not rely on differentiating through equilibria.
\end{example}

\subsection{Proof of \Cref{prop:equilibrium-potential}}\label{append_prop:equilibrium-potential}
\begin{proof}
% (i) The feasible set $\mathcal{C}$ is a nonempty compact convex polytope (a convex hull of
% finitely many points). The function $f$ is continuous and strictly convex on $\mathcal{C}$,
% hence it admits a unique minimizer $y^\star \in \mathcal{C}$.
We first recall the first-order optimality condition for $\min_{y\in\mathcal{C}}f(y)$.
Since $f$ is differentiable and convex, a point $\bar y \in \mathcal{C}$ minimizes $f$ over
$\mathcal{C}$ if and only if
\begin{equation}
\label{eq:VI-condition}
  \langle \nabla f(\bar y),\, y - \bar y \rangle \;\ge\; 0
  \qquad \text{for all } y \in \mathcal{C}.
\end{equation}

\paragraph{(Equilibrium $\Rightarrow$ minimizer).}
Let $z \in \Delta^d$ be a Wardrop equilibrium and $y = y(z)$.
By definition of $f$ we have $\nabla f(y)_i = c_i(y_i)$, so for any $y' \in \mathcal{C}$,
\[
  \langle \nabla f(y),\, y' - y \rangle
  = \sum_{i \in [n]} c_i(y_i)\,(y'_i - y_i).
\]
Every $y' \in \mathcal{C}$ can be written as
$y' = \sum_{S \in \mathcal{S}} \xi_S \,\mathbf{1}_S$ for some $\xi \in \Delta^d$.
Using $y = \sum_{S} z_S \mathbf{1}_S$ we obtain
\begin{align*}
  \langle \nabla f(y),\, y' - y \rangle
  &= \sum_{i \in [n]} c_i(y_i)
     \Bigl( \sum_{S \in \mathcal{S}} \xi_S [\mathbf{1}_S]_i
           - \sum_{S \in \mathcal{S}} z_S [\mathbf{1}_S]_i \Bigr) \\
  &= \sum_{S \in \mathcal{S}} \xi_S \sum_{i \in S} c_i(y_i)
     \;-\; \sum_{S \in \mathcal{S}} z_S \sum_{i \in S} c_i(y_i) \\
  &= \sum_{S \in \mathcal{S}} \xi_S \, c_S(y)
     \;-\; \sum_{S \in \mathcal{S}} z_S \, c_S(y).
\end{align*}
Let $C_{\min} := \min_{S \in \mathcal{S}} c_S(y)$. By the Wardrop condition,
$c_S(y) = C_{\min}$ whenever $z_S > 0$, and therefore
$\sum_{S} z_S c_S(y) = C_{\min}$. On the other hand,
$\sum_{S} \xi_S c_S(y) \ge C_{\min}$ for any $\xi \in \Delta^d$.
Hence
\[
  \langle \nabla f(y),\, y' - y \rangle
  \;=\; \sum_{S} \xi_S c_S(y) - C_{\min}
  \;\ge\; 0
  \qquad\text{for all } y' \in \mathcal{C}.
\]
Thus $y$ satisfies~\eqref{eq:VI-condition} and so $y$ minimizes $f$ over $\mathcal{C}$.
By uniqueness of the minimizer, $y = y^\star$.

\paragraph{(Minimizer $\Rightarrow$ equilibrium).}
Conversely, let $y \in \mathcal{C}$ minimize $f$ over $\mathcal{C}$, and let
$z \in \Delta^d$ be any population profile with $y = y(z)$.
Assume for contradiction that $z$ is not a Wardrop equilibrium.
Then there exist strategies $S,S' \in \mathcal{S}$ such that $z_S > 0$ and
$c_{S'}(y) < c_S(y)$.

For sufficiently small $\varepsilon > 0$ we can define a perturbed profile
$z^\varepsilon \in \Delta^d$ by moving mass~$\varepsilon$ from $S$ to $S'$:
\[
  z^\varepsilon_{S}   = z_S - \varepsilon,\qquad
  z^\varepsilon_{S'} = z_{S'} + \varepsilon,\qquad
  z^\varepsilon_{T}   = z_T \text{ for all } T \notin \{S,S'\},
\]
and let $y^\varepsilon := y(z^\varepsilon) \in \mathcal{C}$ be the induced load.
Then
\[
  y^\varepsilon - y = \varepsilon(\mathbf{1}_{S'} - \mathbf{1}_S).
\]
Using differentiability of $f$ we compute the directional derivative of $f$ at $y$
in the direction $y^\varepsilon - y$:
\begin{align*}
  \left.\frac{d}{d\varepsilon} f(y^\varepsilon)\right|_{\varepsilon = 0}
  &= \left\langle \nabla f(y),\, \mathbf{1}_{S'} - \mathbf{1}_S \right\rangle \\
  &= \sum_{i \in S'} c_i(y_i) - \sum_{i \in S} c_i(y_i)
   \;=\; c_{S'}(y) - c_S(y) \;<\; 0.
\end{align*}
Therefore, for all sufficiently small $\varepsilon > 0$ we have
$f(y^\varepsilon) < f(y)$, contradicting the fact that $y$ minimizes $f$ over $\mathcal{C}$.
Hence no such pair $(S,S')$ can exist, and $z$ must be a Wardrop equilibrium.

Combining the two directions, we obtain the claimed equivalence.
\end{proof}
\subsection{Proof of \Cref{lem_solution_map_Lip}}\label{app:differentiability}
% \begin{lemma}[Lipschitzness of the equilibrium map and hyper-objective]\label{lem_solution_map_Lip}
% Let \Cref{ass:lsm,ass:LL_UL_regular} hold and assume $\Theta$ is convex.
% Then the equilibrium map $\theta\mapsto y^\star(\theta)$ is globally Lipschitz on $\Theta$:
% \begin{equation}\label{eq:y_star_global_lip}
% \|y^\star(\theta)-y^\star(\theta')\| \le \frac{L_{f,2}}{\alpha}\,\|\theta-\theta'\|
% \qquad\forall \theta,\theta'\in\Theta.
% \end{equation}
% Moreover, the hyper-objective $\Phi(\theta)=F(\theta,y^\star(\theta))$ is globally Lipschitz on $\Theta$:
% \begin{equation}\label{eq:phi_global_lip}
% |\Phi(\theta)-\Phi(\theta')|
% \le \Bigl(L_{F,1}+\frac{L_{F,2}L_{f,2}}{\alpha}\Bigr)\|\theta-\theta'\|
% \qquad\forall \theta,\theta'\in\Theta.
% \end{equation}
% \end{lemma}

\begin{proof}
\textbf{Step 1 (a uniform local Lipschitz step).}
Let $\delta:=\alpha r/(2L_{f,2})$. We first show that whenever $\|\theta-\theta'\|\le \delta$,
\begin{equation}\label{eq:local_step}
\|y^\star(\theta)-y^\star(\theta')\|
\le \frac{2L_{f,2}}{\alpha}\,\|\theta-\theta'\|.
\end{equation}
Fix such $\theta,\theta'$ and denote $y_\theta:=y^\star(\theta)$, $y_{\theta'}:=y^\star(\theta')$, and $d:=y_{\theta'}-y_\theta$.
Since $y_{\theta'}$ minimizes $f(\theta',\cdot)$ over $\mathcal C$, the first-order optimality condition can be written as
\begin{equation}\label{eq:kkt_yprime}
0 \in \nabla_y f(\theta',y_{\theta'}) + N_{\mathcal C}(y_{\theta'}),
\end{equation}
where $N_{\mathcal C}(y_{\theta'}):=\{v\in\R^n:\langle v,u-y_{\theta'}\rangle\le 0\ \forall u\in\mathcal C\}$ is the normal cone of $\mathcal C$ at $y_{\theta'}$.
Equivalently, $-\nabla_y f(\theta',y_{\theta'})\in N_{\mathcal C}(y_{\theta'})$.

Consider the stationarity residual of $y_{\theta'}$ for the problem at parameter $\theta$:
\[
\rho
:=\mathrm{dist}\bigl(0,\nabla_y f(\theta,y_{\theta'})+N_{\mathcal C}(y_{\theta'})\bigr)
=\mathrm{dist}\bigl(-\nabla_y f(\theta,y_{\theta'}),N_{\mathcal C}(y_{\theta'})\bigr).
\]
Using $-\nabla_y f(\theta',y_{\theta'})\in N_{\mathcal C}(y_{\theta'})$ and \Cref{ass:LL_UL_regular} (with $y=y_{\theta'}$) gives
\[
\rho
\le \bigl\|\nabla_y f(\theta,y_{\theta'})-\nabla_y f(\theta',y_{\theta'})\bigr\|
\le L_{f,2}\|\theta-\theta'\|.
\]

Next, we use a local error bound consequence of quadratic growth (see, e.g., \citet[Ch.~4--5]{bonnans_shapiro_2000}).
For completeness, we include a short derivation in our setting.
Fix $\theta$ and write $y^\star:=y^\star(\theta)$.
For any $\bar y\in\mathcal C$, define $\rho(\bar y):=\mathrm{dist}\bigl(0,\nabla_y f(\theta,\bar y)+N_{\mathcal C}(\bar y)\bigr)$ and choose $\bar v\in N_{\mathcal C}(\bar y)$ such that $\rho(\bar y)=\|\nabla_y f(\theta,\bar y)+\bar v\|$.
By convexity of $y\mapsto f(\theta,y)$,
$f(\theta,\bar y)-f(\theta,y^\star)\le \langle \nabla_y f(\theta,\bar y),\,\bar y-y^\star\rangle$.
Moreover, since $\bar v\in N_{\mathcal C}(\bar y)$ and $y^\star\in\mathcal C$, we have $\langle \bar v,y^\star-\bar y\rangle\le 0$, i.e., $\langle \bar v,\bar y-y^\star\rangle\ge 0$.
Thus
\[
f(\theta,\bar y)-f(\theta,y^\star)
\le \langle \nabla_y f(\theta,\bar y)+\bar v,\,\bar y-y^\star\rangle
\le \rho(\bar y)\,\|\bar y-y^\star\|.
\]
Combining this with local quadratic growth \eqref{eq:qg} (when $\|\bar y-y^\star\|\le r$) yields: whenever $\|\bar y-y^\star(\theta)\|\le r$,
\begin{equation}\label{eq:qg_kkt_error_bound}
\|\bar y-y^\star(\theta)\|
\le \frac{2}{\alpha}\,\mathrm{dist}\bigl(0,\nabla_y f(\theta,\bar y)+N_{\mathcal C}(\bar y)\bigr)
= \frac{2\rho(\bar y)}{\alpha}.
\end{equation}
Applying \eqref{eq:qg_kkt_error_bound} with $\bar y=y_{\theta'}$ and using $\rho(y_{\theta'})=\rho$ yields
\[
\|d\|=\|y_{\theta'}-y_\theta\|
\le \frac{2\rho}{\alpha}
\le \frac{2L_{f,2}}{\alpha}\|\theta-\theta'\|.
\]
Under $\|\theta-\theta'\|\le \delta=\alpha r/(2L_{f,2})$, this implies $\|d\|\le r$, so the premise of \eqref{eq:qg_kkt_error_bound} holds self-consistently. This proves \eqref{eq:local_step}.

\textbf{Step 2 (global Lipschitzness by chaining).}
Fix arbitrary $\theta,\theta'\in\Theta$. If $\theta=\theta'$ the claim is trivial.
Otherwise set
\[
m := \left\lceil \frac{\|\theta-\theta'\|}{\delta}\right\rceil
\quad\text{and}\quad
\theta_i := \theta + \frac{i}{m}(\theta'-\theta),\ \ i=0,1,\dots,m.
\]
Since $\Theta$ is convex, $\theta_i\in\Theta$ for all $i$, and
\(
\|\theta_{i+1}-\theta_i\| = \|\theta'-\theta\|/m \le \delta.
\)
Applying the local step \eqref{eq:local_step} to each consecutive pair and summing gives
\[
\|y^\star(\theta)-y^\star(\theta')\|
\le \sum_{i=0}^{m-1}\|y^\star(\theta_{i+1})-y^\star(\theta_i)\|
\le \frac{2L_{f,2}}{\alpha}\sum_{i=0}^{m-1}\|\theta_{i+1}-\theta_i\|
= \frac{2L_{f,2}}{\alpha}\|\theta-\theta'\|,
\]
which is \eqref{eq:y_star_global_lip}.

\textbf{Step 3 (Lipschitzness of $\Phi$).}
Using \Cref{ass:LL_UL_regular} for $F$,
\[
|\Phi(\theta)-\Phi(\theta')|
=|F(\theta,y^\star(\theta)) - F(\theta',y^\star(\theta'))|
\le L_{F,1}\|\theta-\theta'\| + L_{F,1}\|y^\star(\theta)-y^\star(\theta')\|.
\]
Plug \eqref{eq:y_star_global_lip} into the last term to obtain \eqref{eq:phi_global_lip}.
\end{proof}

% --- Local growth assumption (needed to convert function gap -> distance) ---

% --- FW inner-loop theorem ---

\section{Proofs for Section~\ref{sec:theory}}\label{app:proofs_sec5}

\subsection{Auxiliary Lemmas}\label{app:aux_lemmas}

Throughout this appendix we fix a parameter $\theta\in\Theta$ and write
$f(y):=f(\theta,y)$, $y^\star:=y^\star(\theta)\in\arg\min_{y\in\mathcal C} f(y)$.
Let $D:=\max\{\|u-v\|:u,v\in\mathcal C\}$.
\begin{lem}[From local quadratic growth to a global error bound]\label{lem:qg_global_eb}
Assume \Cref{ass:lsm}: there exist $\alpha>0$ and $r>0$ such that for all $y\in\mathcal C$
with $\|y-y^\star\|\le r$,
\[
f(y)-f(y^\star)\ge \frac{\alpha}{2}\|y-y^\star\|^2.
\]
If $f$ is convex on $\mathcal C$, then for every $y\in\mathcal C$,
\[
f(y)-f(y^\star)\;\ge\;\frac{\alpha}{2}\min\{\|y-y^\star\|^2,\; r\|y-y^\star\|\}.
\]
In particular, if $f(y)-f(y^\star)\le \frac{\alpha}{2}r^2$, then $\|y-y^\star\|\le r$
and $\|y-y^\star\|^2\le \frac{2}{\alpha}(f(y)-f(y^\star))$.
\end{lem}

\begin{proof}
If $\|y-y^\star\|\le r$ the claim is exactly \Cref{ass:lsm}.
Otherwise let $d:=\|y-y^\star\|>r$ and define the point on the segment
$y_r:=y^\star+\frac{r}{d}(y-y^\star)$ so that $\|y_r-y^\star\|=r$.
By convexity,
\[
f(y_r)\le \Big(1-\frac{r}{d}\Big)f(y^\star)+\frac{r}{d}f(y)
\quad\Rightarrow\quad
f(y)-f(y^\star)\ge \frac{d}{r}\big(f(y_r)-f(y^\star)\big).
\]
Applying \Cref{ass:lsm} to $y_r$ yields $f(y_r)-f(y^\star)\ge \frac{\alpha}{2}r^2$, hence
$f(y)-f(y^\star)\ge \frac{d}{r}\cdot \frac{\alpha}{2}r^2=\frac{\alpha r}{2}d$.
This gives the stated $\min\{\cdot,\cdot\}$ lower bound and the corollaries.
\end{proof}

\subsection{Proof of \Cref{thm:fw_equilibrium_subsampled}}\label{app:proof_fw_subsampled}

\begin{proof}
Fix $\theta$ and abbreviate $f(y):=f(\theta,y)$ and $y^\star:=y^\star(\theta)$.
Let $h_t:=f(y_t)-f(y^\star)$.
By $L_{f,2}$-smoothness and for all $y,s\in\mathcal C$ with $\|s-y\|\le D$ and any $\gamma\in[0,1]$,
\begin{equation}\label{eq:curvature_ineq_clean}
f\big((1-\gamma)y+\gamma s\big)
\le f(y)+\gamma\langle \nabla f(y),s-y\rangle+\frac{C_f}{2}\gamma^2,
\qquad \text{where } C_f:=L_{f,2}D^2.
\end{equation}
Let $s^\star_t\in\arg\min_{s\in\mathcal C}\langle \nabla f(y_t),s\rangle$ and define the (exact) FW gap
\[
g_t:=\langle \nabla f(y_t),y_t-s_t^\star\rangle.
\]
By convexity, $h_t\le g_t$ for all $t$.

\paragraph{Step 1: the hit event.}
Let $\mathcal H_t$ be the event that the $m$ i.i.d.\ samples used by $\mathrm{LMO}_m$
contain an exact minimizer of $s\mapsto\langle \nabla f(y_t),s\rangle$.
On $\mathcal H_t$, the best-of-$m$ rule returns an exact minimizer, hence $s_t=s_t^\star$.
By \Cref{ass:opt_mass}, conditionally on $y_t$,
\begin{equation}\label{eq:hit_prob_clean}
\Pr(\mathcal H_t\mid y_t)\ge \kappa_m.
\end{equation}

\paragraph{Comparison to \citet{kerdreux2018fw_subsampling}.}
At a high level, the argument mirrors the subsampling-oracle FW analysis of \citet{kerdreux2018fw_subsampling} (see their proof of Theorem~2.1): identify a per-iteration event on which the update coincides with exact FW, and then average the resulting decrease over this event.
The main additional challenge in our setting is that we do \emph{not} assume a uniform-inclusion model.
In \citet{kerdreux2018fw_subsampling}, every atom is included with the same probability $\eta$, so the probability of hitting an exact LMO minimizer is controlled in a way that is essentially independent of the current iterate.
Here, we sample candidates from an arbitrary distribution $q$ over $\mathcal S$, so the hit probability depends on the current FW gradient through $\mathrm{Opt}(g)$; in particular, it is coupled with the algorithmic trajectory.
To make this dependence explicit (and keep the proof clean), we isolate $p(g)=q(\mathrm{Opt}(g))$ and its $m$-sample amplification $\kappa_m(g)=1-(1-p(g))^m$, and we assume a uniform lower bound along the run (Assumption~\ref{ass:opt_mass}).
Given this, the proof proceeds by conditioning on $y_t$ and combining (i) an exact-FW decrease on $\mathcal H_t$ with (ii) monotonicity on $\mathcal H_t^c$ (via short-step or exact line search), which yields the conditional recursion \eqref{eq:cond_rec_psi_clean} with the constant $\kappa_m$ playing the role of $\eta$.
This pinpoints precisely where general, structure-aware sampling enters the analysis and why improving $\kappa_m$ (e.g., via stratification) improves the effective rate.

\paragraph{Step 2: approximate line search and definition of $\psi$.}
We analyze a smoothness-based \emph{short-step} rule (an approximate line search):
define the directional gap for the chosen atom $s_t$,
\[
\widetilde g_t:=\langle \nabla f(y_t),y_t-s_t\rangle,
\]
and take
\begin{equation}\label{eq:short_step_gamma}
\gamma_t
\in
\argmax_{\gamma\in[0,1]}\Big\{\gamma\,\widetilde g_t-\frac{C_f}{2}\gamma^2\Big\}
\quad\Longleftrightarrow\quad
\gamma_t=\min\Big\{\max\Big\{\frac{\widetilde g_t}{C_f},\,0\Big\},\,1\Big\}.
\end{equation}
Applying \eqref{eq:curvature_ineq_clean} with $y=y_t$, $s=s_t$, and $\gamma=\gamma_t$ yields
\[
f(y_{t+1})
\le
f(y_t)-\gamma_t\,\widetilde g_t+\frac{C_f}{2}\gamma_t^2
\le
f(y_t),
\]
and hence $h_{t+1}\le h_t$ always.

On the hit event $\mathcal H_t$ we have $s_t=s_t^\star$ and thus $\widetilde g_t=g_t$.
Define
\begin{equation}\label{eq:def_psi}
\psi(g):=\max_{\gamma\in[0,1]}\Big\{\gamma g-\frac{C_f}{2}\gamma^2\Big\}.
\end{equation}
By the choice of $\gamma_t$ in \eqref{eq:short_step_gamma}, we have
$\gamma_t\,\widetilde g_t-\frac{C_f}{2}\gamma_t^2=\psi(\widetilde g_t)$.
Then, on $\mathcal H_t$,
\begin{equation}\label{eq:hit_decrease_psi}
h_{t+1}\le h_t-\psi(g_t).
\end{equation}
If instead we use \emph{exact} line search, the same bound holds since the exact minimizer over $\gamma\in[0,1]$ achieves an objective value no larger than the short-step choice in \eqref{eq:short_step_gamma}.
Moreover, \eqref{eq:def_psi} has the closed form
\begin{equation}\label{eq:psi_closed_form}
\psi(g)=
\begin{cases}
\frac{g^2}{2C_f}, & 0\le g\le C_f,\\[4pt]
g-\frac{C_f}{2}, & g\ge C_f,
\end{cases}
\end{equation}
since the maximizer is $\gamma^\star(g)=\min\{g/C_f,\,1\}$.
In particular, $\psi(g)\ge 0$ for all $g\ge 0$.

\paragraph{Step 3: conditional recursion.}
From \eqref{eq:def_psi}, for each fixed $\gamma\in[0,1]$ the map
$g\mapsto \gamma g-\frac{C_f}{2}\gamma^2$ is affine and nondecreasing in $g$.
Hence $\psi$, being the pointwise supremum over $\gamma\in[0,1]$ of such functions, is
(i) \emph{nondecreasing} and (ii) \emph{convex} on $[0,\infty)$.

On $\mathcal H_t^c$ we have $h_{t+1}\le h_t$. Using \eqref{eq:hit_decrease_psi} on $\mathcal H_t$ and conditioning on $y_t$,
\begin{align}
\E[h_{t+1}\mid y_t]
&\le \Pr(\mathcal H_t\mid y_t)\,(h_t-\psi(g_t))+\Pr(\mathcal H_t^c\mid y_t)\,h_t \nonumber\\
&= h_t-\Pr(\mathcal H_t\mid y_t)\,\psi(g_t)
\;\le\; h_t-\kappa_m\,\psi(g_t), \label{eq:cond_rec_psi_clean}
\end{align}
where we used \eqref{eq:hit_prob_clean}.
Since $g_t\ge h_t$ and $\psi$ is nondecreasing, $\psi(g_t)\ge \psi(h_t)$, so
\begin{equation}\label{eq:cond_rec_h_clean}
\E[h_{t+1}\mid y_t]\le h_t-\kappa_m\,\psi(h_t).
\end{equation}

Take total expectation in \eqref{eq:cond_rec_h_clean}:
\[
\E[h_{t+1}] \le \E[h_t]-\kappa_m\,\E[\psi(h_t)].
\]
Since $\psi$ is convex, Jensen gives $\E[\psi(h_t)]\ge \psi(\E[h_t])$.
Let $a_t:=\E[h_t]$. Then
\begin{equation}\label{eq:rec_a_clean}
a_{t+1}\le a_t-\kappa_m\,\psi(a_t).
\end{equation}

\paragraph{Step 4: solve the recursion.}
Using the closed form \eqref{eq:psi_closed_form}:

\smallskip
\noindent\emph{Regime 1: $a_t\ge C_f$.}
Then $\psi(a_t)=a_t-\frac{C_f}{2}$ and \eqref{eq:rec_a_clean} yields
\[
a_{t+1}\le (1-\kappa_m)a_t+\kappa_m\frac{C_f}{2}.
\]
Iterating gives
\begin{equation}\label{eq:geom_bound_clean}
a_t \le (1-\kappa_m)^t a_0 + \frac{C_f}{2}
\qquad \forall t\ge 0.
\end{equation}

\smallskip
\noindent\emph{Regime 2: $a_t\le C_f$.}
Then $\psi(a_t)=\frac{a_t^2}{2C_f}$ and \eqref{eq:rec_a_clean} becomes
\begin{equation}\label{eq:quad_rec_clean}
a_{t+1}\le a_t-\frac{\kappa_m}{2C_f}a_t^2.
\end{equation}
Let $c:=\kappa_m/(2C_f)$. Since $a_{t+1}\le a_t$, we may assume $a_t>0$.
From \eqref{eq:quad_rec_clean} we have $a_{t+1}\le a_t(1-c a_t)$, hence
\[
\frac{1}{a_{t+1}} \ge \frac{1}{a_t(1-c a_t)}.
\]
Using $a_t\le C_f$ gives $c a_t \le \kappa_m/2 < 1$, and thus $\frac{1}{1-x}\ge 1+x$ for $x\in[0,1)$ implies
\[
\frac{1}{a_{t+1}}
\ge \frac{1}{a_t}\Big(1+c a_t\Big)
= \frac{1}{a_t}+c.
\]
Summing from $t=0$ to $T-1$ yields $\frac{1}{a_T}\ge \frac{1}{a_0}+cT\ge cT$, hence
\begin{equation}\label{eq:sublinear_bound_clean}
a_T\le \frac{1}{cT}=\frac{2C_f}{\kappa_m T}.
\end{equation}
From Regime~1 we have for all $t\ge 0$ the geometric bound
\[
a_t \le (1-\kappa_m)^t a_0 + \frac{C_f}{2}. \tag{\ref{eq:geom_bound_clean}}
\]
Choose $T_0$ as in \eqref{eq:burn_in_def}. Using $1-x\le e^{-x}$ for $x\in[0,1]$ gives
\[
(1-\kappa_m)^{T_0}a_0 \le e^{-\kappa_m T_0}a_0 \le \frac{C_f}{2},
\]
and therefore
\[
a_{T_0} \le (1-\kappa_m)^{T_0}a_0+\frac{C_f}{2}\le C_f.
\]
Hence, for all $t\ge T_0$ the iterates satisfy $a_t\le C_f$, so Regime~2 applies from time $T_0$ onward.
Applying the sublinear estimate \eqref{eq:sublinear_bound_clean} to the shifted sequence
$\{a_{T_0},a_{T_0+1},\dots\}$ yields, for all $T\ge T_0$,
\[
a_T \le \frac{2C_f}{\kappa_m\,(T-T_0+1)}.
\]
Recalling $a_T=\E[h_T]=\E[f(\theta,y_T)-f(\theta,y^\star(\theta))]$ proves \eqref{eq:fw_linesearch_rate_kappa_burnin}.

\paragraph{Step 5: \emph{mean-square} distance bound under \Cref{ass:lsm}.}
Let \Cref{ass:lsm} holds. Fix any $\varepsilon\in(0,r]$ and $T\ge T_0$.
Let $d_T:=\|y_T-y^\star\|$. Since $y_T,y^\star\in\mathcal C$, we have $d_T\le D$ almost surely.
By \Cref{lem:qg_global_eb}, $h_T:=f(y_T)-f(y^\star)\ge \frac{\alpha}{2}\min\{d_T^2,\,r d_T\}$.
If $d_T\le r$, then $d_T^2\le \frac{2}{\alpha}h_T$.
If $d_T>r$, then $h_T\ge \frac{\alpha r}{2}d_T$, so $d_T\le \frac{2}{\alpha r}h_T$ and hence $d_T^2\le D d_T\le \frac{2D}{\alpha r}h_T$.
Thus,
\[
\|y_T-y^\star\|^2 \;=\; d_T^2 \;\le\; \frac{2}{\alpha}\max\Big\{1,\frac{D}{r}\Big\}\,h_T
\qquad\text{almost surely.}
\]
Taking expectations gives
\begin{equation}\label{eq:ms_dist_from_gap}
\mathbb E\|y_T-y^\star\|^2 \;\le\; \frac{2}{\alpha}\max\Big\{1,\frac{D}{r}\Big\}\,\mathbb E[h_T].
\end{equation}
Using \eqref{eq:fw_linesearch_rate_kappa_burnin} and $C_f=L_{f,2}D^2$, for all $T\ge T_0$ we have
\[
\mathbb E[h_T]
=
\mathbb E[f(y_T)]-f(y^\star)
\le \frac{2C_f}{\kappa_m\,(T-T_0+1)}.
\]
Combining with \eqref{eq:ms_dist_from_gap} gives
\begin{equation}\label{eq:ms_dist_rate}
\mathbb E\|y_T-y^\star\|^2
\le
\frac{4C_f}{\alpha\,\kappa_m\,(T-T_0+1)}\max\Big\{1,\frac{D}{r}\Big\}
=
\frac{4L_{f,2}D^2}{\alpha\,\kappa_m\,(T-T_0+1)}\max\Big\{1,\frac{D}{r}\Big\}.
\end{equation}
Therefore, it suffices to choose $T$ so that the right-hand side of \eqref{eq:ms_dist_rate} is at most $\varepsilon^2$, i.e.,
\[
T-T_0+1 \;\ge\; \frac{4L_{f,2}D^2}{\alpha\,\kappa_m\,\varepsilon^2}\max\Big\{1,\frac{D}{r}\Big\},
\qquad\text{equivalently}\qquad
T \;\ge\; \left\lceil\frac{4L_{f,2}D^2}{\alpha\,\kappa_m\,\varepsilon^2}\max\Big\{1,\frac{D}{r}\Big\}-1\right\rceil + T_0.
\]
If $\alpha,r$ are uniform over $\Theta$, then the bound holds uniformly in $\theta$, hence
\[
\sup_{\theta\in\Theta}\mathbb E\|y_T(\theta)-y^\star(\theta)\|^2 \le \varepsilon^2,
\]
as claimed.

\end{proof}

\subsection{Proof of \Cref{thm:algo2_ggsp_phi_subsampled}}\label{app:proof_zo_ggsp}

\begin{proof}
This proof is a standard ``randomized smoothing + projected (stochastic) gradient mapping" argument;
see, e.g., \citep{nesterov2017random,ghadimi2013stochastic,liu2024zo_constrained_nonsmooth}.
We include it for completeness and identify the only place where the inner-loop error enters.
\paragraph{Step 1: smoothing and Goldstein subdifferentials.}
Define the ball-smoothed function
\[
\Phi_\rho(\theta):=\mathbb E_{v\sim\mathrm{Unif}(\mathbb B)}[\Phi(\theta+\rho v)].
\]
This is the standard uniform-ball smoothing used in gradient-free methods \citep{nesterov2017random,ghadimi2013stochastic}.
The following facts are standard for Lipschitz $\Phi$:
\begin{enumerate}
\item (Approximation) $|\Phi_\rho(\theta)-\Phi(\theta)|\le \rho L_\Phi$ for all $\theta$ \citep{nesterov2017random}.
\item (Differentiability) $\Phi_\rho$ is continuously differentiable, and $\nabla \Phi_\rho$ is Lipschitz with constant
$L_\rho\le c\,\frac{L_\Phi\sqrt{k}}{\rho}$ for a universal constant $c>0$ \citep{nesterov2017random}.
\item (Goldstein inclusion) $\nabla \Phi_\rho(\theta)\in \partial_\rho \Phi(\theta)$ for all $\theta$ \citep{goldstein1977optimization}.
\end{enumerate}
We use these as black-box smoothing properties.

\paragraph{Step 2: inexact value oracle.}
At iteration $t$, sample $u_{t,1},\dots,u_{t,B}\overset{\text{i.i.d.}}{\sim}\mathrm{Unif}(\mathbb S^{k-1})$.
Define the \emph{mini-batched} symmetric estimator (the algorithm uses $\widehat \Phi_T(\theta)=F(\theta,y_T(\theta))$):
\begin{align*}
\widehat g_t
&:=\frac{1}{B}\sum_{i=1}^B \widehat g_{t,i},
\qquad
\widehat g_{t,i}
:=\frac{k}{2\rho}\Big(\widehat\Phi_T(\theta_t+\rho u_{t,i})-\widehat\Phi_T(\theta_t-\rho u_{t,i})\Big)u_{t,i}.
\end{align*}
Define also the corresponding ideal (exact-hyperobjective) estimators
\begin{align*}
g_t
&:=\frac{1}{B}\sum_{i=1}^B g_{t,i},
\qquad
g_{t,i}
:=\frac{k}{2\rho}\Big(\Phi(\theta_t+\rho u_{t,i})-\Phi(\theta_t-\rho u_{t,i})\Big)u_{t,i}.
\end{align*}
Then (standard) $\mathbb E[g_{t,i}\mid \theta_t]=\nabla \Phi_\rho(\theta_t)$ and hence $\mathbb E[g_t\mid \theta_t]=\nabla \Phi_\rho(\theta_t)$
\citep{nesterov2017random,ghadimi2013stochastic}. Moreover, since $g_t$ is an average of $B$ i.i.d.\ copies,
\begin{equation}\label{eq:gt_var_batch}
\mathbb E\|g_t-\nabla \Phi_\rho(\theta_t)\|^2
=\mathbb E\big[\mathrm{Var}(g_t\mid \theta_t)\big]
=\frac{1}{B}\,\mathbb E\big[\mathrm{Var}(g_{t,1}\mid \theta_t)\big]
\le \frac{1}{B}\,\mathbb E\|g_{t,1}\|^2
\le \frac{c'kL_\Phi^2}{B},
\end{equation}
using the standard second-moment bound $\mathbb E\|g_{t,1}\|^2\le c'kL_\Phi^2$ for a universal $c'>0$.

Now define the (per-direction) pointwise oracle errors
\[
e^+_{t,i}:=\widehat\Phi_T(\theta_t+\rho u_{t,i})-\Phi(\theta_t+\rho u_{t,i}),
\qquad
e^-_{t,i}:=\widehat\Phi_T(\theta_t-\rho u_{t,i})-\Phi(\theta_t-\rho u_{t,i}).
\]
By $L_{F,1}$-Lipschitzness of $F(\theta,\cdot)$ and the assumed mean-square inner accuracy
$\mathbb E\|y_T(\theta)-y^\star(\theta)\|^2\le \varepsilon_y^2$ (for all queried $\theta$),
\begin{equation}\label{eq:e_sq_bound}
\mathbb E[(e^\pm_{t,i})^2]
\le L_{F,1}^2\,\mathbb E\|y_T(\theta_t\pm \rho u_{t,i})-y^\star(\theta_t\pm \rho u_{t,i})\|^2
\le L_{F,1}^2\,\varepsilon_y^2.
\end{equation}

\paragraph{Step 3: projected descent on the smoothed objective.}
Let $\theta_{t+1}=\Pi_\Theta(\theta_t-\eta \widehat g_t)$ and define the gradient mapping
\[
\widehat G_t:=G_\Theta(\theta_t,\widehat g_t;\eta)
=\frac{1}{\eta}\Big(\theta_t-\Pi_\Theta(\theta_t-\eta\widehat g_t)\Big),
\qquad\text{so}\qquad
\theta_{t+1}-\theta_t=-\eta \widehat G_t.
\]
Since $\Phi_\rho$ is $L_\rho$-smooth, the descent lemma gives
\[
\Phi_\rho(\theta_{t+1})
\le \Phi_\rho(\theta_t)+\langle \nabla \Phi_\rho(\theta_t),\theta_{t+1}-\theta_t\rangle+\frac{L_\rho}{2}\|\theta_{t+1}-\theta_t\|^2
= \Phi_\rho(\theta_t)-\eta\langle \nabla \Phi_\rho(\theta_t),\widehat G_t\rangle+\frac{L_\rho\eta^2}{2}\|\widehat G_t\|^2.
\]
Moreover, the projection inequality for $\theta_{t+1}=\Pi_\Theta(\theta_t-\eta\widehat g_t)$ yields
$\langle \widehat g_t,\widehat G_t\rangle\ge \|\widehat G_t\|^2$.
Therefore,
\[
-\eta\langle \nabla \Phi_\rho(\theta_t),\widehat G_t\rangle
= -\eta\langle \widehat g_t,\widehat G_t\rangle
+\eta\langle \widehat G_t,\widehat g_t-\nabla \Phi_\rho(\theta_t)\rangle
\le -\eta\|\widehat G_t\|^2+\eta\langle \widehat G_t,\widehat g_t-\nabla \Phi_\rho(\theta_t)\rangle.
\]
If $\eta\le 1/L_\rho$, then
$-\eta\|\widehat G_t\|^2+\frac{L_\rho\eta^2}{2}\|\widehat G_t\|^2\le -\frac{\eta}{2}\|\widehat G_t\|^2$,
hence
\[
\Phi_\rho(\theta_{t+1})
\le \Phi_\rho(\theta_t)
-\frac{\eta}{2}\|\widehat G_t\|^2
+\eta\langle \widehat G_t,\widehat g_t-\nabla \Phi_\rho(\theta_t)\rangle.
\]
Applying Young's inequality $\langle a,b\rangle\le \frac14\|a\|^2+\|b\|^2$ gives
\begin{equation}\label{eq:one_step_descent_batch}
\Phi_\rho(\theta_{t+1})
\le \Phi_\rho(\theta_t)-\frac{\eta}{4}\|\widehat G_t\|^2+\eta\|\widehat g_t-\nabla \Phi_\rho(\theta_t)\|^2.
\end{equation}

\paragraph{Step 4: mean-square control of $\|\widehat g_t-\nabla \Phi_\rho(\theta_t)\|^2$.}
Let $\delta_t:=\widehat g_t-\nabla \Phi_\rho(\theta_t)$.
Using $\delta_t=(\widehat g_t-g_t)+(g_t-\nabla \Phi_\rho(\theta_t))$ and $\|a+b\|^2\le 2\|a\|^2+2\|b\|^2$,
\begin{equation}\label{eq:delta_split_batch}
\mathbb E\|\delta_t\|^2
\le 2\,\mathbb E\|\widehat g_t-g_t\|^2
+2\,\mathbb E\|g_t-\nabla \Phi_\rho(\theta_t)\|^2.
\end{equation}

\emph{(i) Inexact-oracle term $\mathbb E\|\widehat g_t-g_t\|^2$.}
For each $i$,
\[
\widehat g_{t,i}-g_{t,i}
=\frac{k}{2\rho}(e^+_{t,i}-e^-_{t,i})u_{t,i},
\qquad \|u_{t,i}\|=1,
\]
so
\[
\|\widehat g_{t,i}-g_{t,i}\|^2
=\Big(\frac{k}{2\rho}\Big)^2(e^+_{t,i}-e^-_{t,i})^2
\le \Big(\frac{k}{2\rho}\Big)^2\cdot 2\big((e^+_{t,i})^2+(e^-_{t,i})^2\big).
\]
Taking expectations and using \eqref{eq:e_sq_bound} yields
\begin{equation}\label{eq:single_oracle_sq}
\mathbb E\|\widehat g_{t,1}-g_{t,1}\|^2
\le \Big(\frac{k}{2\rho}\Big)^2 \cdot 2\cdot 2 L_{F,1}^2\varepsilon_y^2
= \Big(\frac{k}{\rho}\Big)^2 L_{F,1}^2\varepsilon_y^2.
\end{equation}
Since $\widehat g_t-g_t=\frac{1}{B}\sum_{i=1}^B(\widehat g_{t,i}-g_{t,i})$, Jensen's inequality gives
\begin{equation}\label{eq:oracle_sq_batch}
\mathbb E\|\widehat g_t-g_t\|^2\le \frac{1}{B}\sum_{i=1}^{B}\mathbb{E}[\|\widehat{g}_{t,i}-g_{t,i}\|^{2}]
\le \Big(\frac{k}{\rho}\Big)^2 L_{F,1}^2\varepsilon_y^2.
\end{equation}

\emph{(ii) Intrinsic ZO variance term.}
This is exactly \eqref{eq:gt_var_batch}:
\[
\mathbb E\|g_t-\nabla \Phi_\rho(\theta_t)\|^2\le \frac{c'kL_\Phi^2}{B}.
\]
Combining with \eqref{eq:delta_split_batch}--\eqref{eq:oracle_sq_batch} yields
\begin{equation}\label{eq:delta_sq_final_batch}
\mathbb E\|\widehat g_t-\nabla \Phi_\rho(\theta_t)\|^2
\le 2\Big(\frac{k}{\rho}\Big)^2 L_{F,1}^2\varepsilon_y^2
\;+\;\frac{2c'kL_\Phi^2}{B}.
\end{equation}

\paragraph{Step 5: summation and random iterate.}
Summing \eqref{eq:one_step_descent_batch} over $t=0,\dots,K-1$ and using $\Phi_\rho(\theta_K)\ge \inf_\Theta \Phi_\rho$ gives
\begin{equation}\label{eq:sum_Ghat_batch}
\frac{\eta}{4}\sum_{t=0}^{K-1}\mathbb E\|\widehat G_t\|^2
\le \Phi_\rho(\theta_0)-\inf_\Theta \Phi_\rho+\eta\sum_{t=0}^{K-1}\mathbb E\|\delta_t\|^2.
\end{equation}
Let $\tau$ be uniform on $\{0,\dots,K-1\}$ and set $\widehat\theta:=\theta_\tau$.
Dividing \eqref{eq:sum_Ghat_batch} by $\eta K/4$ yields
\begin{equation}\label{eq:E_Ghat_tau_batch}
\mathbb E\|\widehat G_\tau\|^2
\le \frac{4\,(\Phi_\rho(\theta_0)-\inf_\Theta \Phi_\rho)}{\eta K}
+4\,\mathbb E\|\delta_\tau\|^2.
\end{equation}

Next, by nonexpansiveness of projection,
\[
\|G_\Theta(\theta_t,\nabla \Phi_\rho(\theta_t);\eta)-\widehat G_t\|
=\frac{1}{\eta}\Big\|\Pi_\Theta(\theta_t-\eta\nabla \Phi_\rho(\theta_t))-\Pi_\Theta(\theta_t-\eta\widehat g_t)\Big\|
\le \|\delta_t\|.
\]
Thus $(a+b)^2\le 2a^2+2b^2$ implies
\begin{equation}\label{eq:E_G_tau_sq_batch}
\mathbb E\|G_\Theta(\widehat\theta,\nabla \Phi_\rho(\widehat\theta);\eta)\|^2
\le 2\,\mathbb E\|\widehat G_\tau\|^2+2\,\mathbb E\|\delta_\tau\|^2.
\end{equation}
Combining \eqref{eq:E_Ghat_tau_batch} and \eqref{eq:E_G_tau_sq_batch} gives
\begin{equation}\label{eq:E_G_tau_sq_combined_batch}
\mathbb E\|G_\Theta(\widehat\theta,\nabla \Phi_\rho(\widehat\theta);\eta)\|^2
\le \frac{8\,(\Phi_\rho(\theta_0)-\inf_\Theta \Phi_\rho)}{\eta K}
+10\,\mathbb E\|\delta_\tau\|^2.
\end{equation}

\paragraph{Step 6: plug in Step 4.}
Plugging \eqref{eq:delta_sq_final_batch} into \eqref{eq:E_G_tau_sq_combined_batch} yields universal constants
$C_0,C_0',C_0''>0$ such that
\[
\mathbb E\|G_\Theta(\widehat\theta,\nabla \Phi_\rho(\widehat\theta);\eta)\|^2
\le
C_0\,\frac{\Phi_\rho(\theta_0)-\inf_\Theta \Phi_\rho}{\eta K}
\;+\;
C_0'\Big(\frac{k}{\rho}\Big)^2 L_{F,1}^2\,\varepsilon_y^2
\;+\;
\frac{C_0''}{B}\,kL_\Phi^2.
\]
By Jensen, $\mathbb E\|Z\|\le \sqrt{\mathbb E\|Z\|^2}$, and using
$\Phi_\rho(\theta_0)-\inf_\Theta\Phi_\rho\le \Phi(\theta_0)-\inf_\Theta\Phi+2\rho L_\Phi$,
we obtain constants $C_1,C_2,C_3>0$ such that
\begin{equation}\label{eq:GGSP_batch_preabsorb}
\mathbb E\|G_\Theta(\widehat\theta,\nabla \Phi_\rho(\widehat\theta);\eta)\|
\le
C_1\sqrt{\frac{\Phi(\theta_0)-\inf_\Theta\Phi+2\rho L_\Phi}{\eta K}}
\;+\;
C_2\frac{k}{\rho}\,L_{F,1}\,\varepsilon_y
\;+\;
C_3\sqrt{\frac{k}{B}}\,L_\Phi.
\end{equation}
\end{proof}

% \subsection{Exact LMO: a standard Frank--Wolfe bound}\label{app:fw_exact_lmo}

% \begin{prop}[FW-Equilibrium with an exact LMO]\label{prop:fw_equilibrium_exact}
% Let \Cref{ass:costs,ass:LL_UL_regular} hold. Fix $\theta\in\Theta$ and initialize $y_0\in\mathcal C$.
% Run \Cref{alg:fw-inner} with the exact oracle $\mathrm{LMO}$ and exact line search (or the short-step rule \eqref{eq:short_step_gamma}).
% Then the bounds in \Cref{thm:fw_equilibrium_subsampled} hold deterministically with $\kappa_m=1$ and without requiring \Cref{ass:opt_mass}.
% \end{prop}

% \begin{proof}
% This is the standard Frank--Wolfe analysis (e.g., \citet{frank1956algorithm,jaggi2013revisiting}).
% It also follows by specializing the proof of \Cref{thm:fw_equilibrium_subsampled}: with an exact oracle, the hit event $\mathcal H_t$ holds deterministically (equivalently, take $\kappa_m=1$), so the recursion \eqref{eq:cond_rec_h_clean} holds without conditioning/expectations, and the rest of the argument is identical.
% \end{proof}

%%%%%%%%%%%%%%%%%%%%%%%%%%%%%%%%%%%%%%%%%%%%%%%%%%%%%%%%%%%%%%%%%%%%%%%%%%%%%%%
%%%%%%%%%%%%%%%%%%%%%%%%%%%%%%%%%%%%%%%%%%%%%%%%%%%%%%%%%%%%%%%%%%%%%%%%%%%%%%%

\section{Zero-Suppressed Decision Diagrams (ZDD)}
\label{app:zdd}
% This appendix is included from multiple top-level TeX files.
\IfFileExists{paper/zdd.tex}{Zero-suppressed decision diagrams (ZDDs) are a canonical data structure for representing a family of subsets of a finite universe in a compact form.
They can be viewed as a specialization of reduced ordered binary decision diagrams (BDDs) \citep{bryant1986graph} tailored to sparse set systems \citep{minato_1993_zdd}.
In our setting, the universe typically consists of graph resources (e.g., edges), and a feasible strategy (path, tree, cycle) is encoded by its incidence vector, i.e., a subset of universe items.

\subsection{Definition and Construction}
\paragraph{ZDDs as compressed set representations.}
Fix an ordered universe $U=(e_1,\dots,e_n)$ and a family $\mathcal S\subseteq 2^{U}$.
A ZDD for $\mathcal S$ is a rooted directed acyclic graph with two terminals $\top$ and $\bot$.
Each internal node $v$ is labeled by an item index $\ell(v)\in\{1,\dots,n\}$ and has two outgoing arcs: a \emph{lo}-arc corresponding to excluding $e_{\ell(v)}$ and a \emph{hi}-arc corresponding to including $e_{\ell(v)}$.
A root-to-$\top$ path encodes a unique subset $S\subseteq U$ given by the labels whose hi-arcs are taken; the family represented by the diagram is exactly the collection of subsets obtained from all root-to-$\top$ paths. \Cref{fig:zdd-toy} illustrates this correspondence on a toy network.

The key advantage of ZDDs is that they admit aggressive reductions.
As in BDDs, isomorphic subgraphs are merged so that identical subproblems are shared \citep{bryant1986graph}.
In addition, the \emph{zero-suppression} rule removes any node whose hi-child is $\bot$, redirecting the incoming arc(s) to its lo-child \citep{minato_1993_zdd}.
Intuitively, if including an item can never lead to a feasible set, then that item is irrelevant for the family and should be skipped.
After applying these reductions, the resulting diagram is canonical for a fixed variable ordering of $U$.

\paragraph{Construction and size.}
ZDDs can be built by incremental, constraint-driven enumeration procedures that traverse items in the chosen order and maintain a compact representation of partial feasibility.
For graph families (paths, Steiner structures, cycles), such constructions often proceed by a frontier-style dynamic program that tracks only local connectivity information at the boundary between processed and unprocessed edges.
The resulting preprocessing time is roughly linear in the number of ZDD nodes produced.

The practical challenge is that ZDD size is highly instance dependent.
In the worst case, the number of nodes can grow exponentially with $n$ (and it depends sensitively on the chosen ordering), which makes ZDD-based approaches unsuitable as a universal replacement for combinatorial optimization.
However, when the underlying constraints have exploitable structure, ZDDs can be orders of magnitude smaller than the explicit family $\mathcal S$ and enable exact computations over exponentially large strategy sets \citep{minato_1993_zdd}.

\begin{figure}[t]
\centering
\begin{tikzpicture}[>=Latex, node distance=10mm, every node/.style={font=\small}]
  % ---- Graph (left) ----
  \node[circle, draw, inner sep=1.2pt] (s) {$s$};
  \node[circle, draw, inner sep=1.2pt, right=18mm of s] (u) {$u$};
  \node[circle, draw, inner sep=1.2pt, below=12mm of u] (t) {$t$};

  \draw[->, thick] (s) -- node[above] {$e_1$} (u);
  \draw[->, thick] (u) -- node[right] {$e_2$} (t);
  \draw[->, thick] (s) -- node[left] {$e_3$} (t);

  \node[align=center, font=\small, below=6mm of t] {\strut $s$--$t$ paths: $\{\{e_1,e_2\},\{e_3\}\}$};

  % ---- ZDD (right) ----
  \begin{scope}[xshift=70mm, yshift=0mm]
    \tikzset{znode/.style={circle, draw, inner sep=1.6pt}}
    \tikzset{term/.style={rectangle, draw, inner sep=2pt}}

    \node[znode] (v1) {$e_1$};
    \node[znode, below left=10mm and 10mm of v1] (v3) {$e_3$};
    \node[znode, below right=10mm and 10mm of v1] (v2) {$e_2$};

    \node[term, below=14mm of v2] (top) {$\top$};
    \node[term, below=14mm of v3] (bot) {$\bot$};

    % root splits
    \draw[dashed] (v1) -- node[left] {lo} (v3);
    \draw[thick]  (v1) -- node[right] {hi} (v2);

    % e3 node
    \draw[dashed] (v3) -- node[left] {lo} (bot);
    \draw[thick]  (v3) -- node[right] {hi} (top);

    % e2 node
    \draw[dashed] (v2) -- node[left] {lo} (bot);
    \draw[thick]  (v2) -- node[right] {hi} (top);

    \node[font=\small, above=3mm of v1] {ZDD (order $e_1,e_2,e_3$)};
  \end{scope}
\end{tikzpicture}
\caption{Left: a three-edge network with two $s$--$t$ paths. Right: a ZDD encoding the corresponding strategy family. Root-to-$\top$ paths correspond to feasible strategies, with hi-arcs indicating selected edges.}
\label{fig:zdd-toy}
\end{figure}
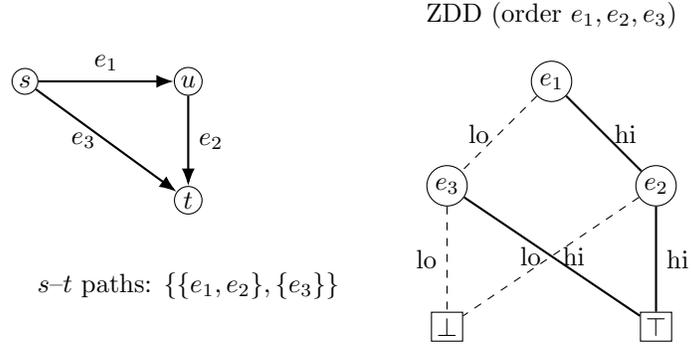

\subsection{Exact Optimization over a ZDD}
A ZDD enables exact minimization of any additive weight function over the underlying family.
Given weights $w\in\R^n$ (e.g., edge costs), consider
\[
  \min_{S\in\mathcal S}\ \sum_{e_i\in S} w_i.
\]
Because each feasible strategy corresponds to a root-to-$\top$ path, this problem is equivalent to finding a minimum-cost root-to-$\top$ path in the ZDD DAG, where taking the hi-arc at a node labeled $i$ incurs cost $w_i$ and taking the lo-arc incurs cost $0$.
Since the ZDD is acyclic, the optimum can be computed by a single bottom-up dynamic program.
Writing $\mathrm{lo}(v)$ and $\mathrm{hi}(v)$ for the children of a node $v$ and letting $\mathrm{cost}(\top)=0$, $\mathrm{cost}(\bot)=+\infty$, we compute
\[
  \mathrm{cost}(v)
  = \min\{\mathrm{cost}(\mathrm{lo}(v)),\; w_{\ell(v)} + \mathrm{cost}(\mathrm{hi}(v))\},
\]
in reverse topological order.
Recovering an optimal strategy amounts to following the minimizing branch from the root.
The runtime is linear in the number of ZDD nodes (and requires only additions and minima), so the diagram acts as an \emph{exact linear minimization oracle} over an exponentially large family.

\subsection{Sampling Strategies from a ZDD}
\label{app:zdd_sampling}
Beyond exact optimization, ZDDs support efficient sampling of feasible strategies without enumerating $\mathcal S$.
This is useful when exact optimization becomes costly (e.g., for very large diagrams) and one wishes to form a subsampled candidate set.
Sampling is again implemented by dynamic programming followed by a randomized root-to-$\top$ traversal.

\paragraph{Uniform sampling over strategies.}
Let $N(v)$ denote the number of root-to-$\top$ completions from a node $v$.
These counts satisfy $N(\top)=1$, $N(\bot)=0$, and
\[
  N(v)=N(\mathrm{lo}(v))+N(\mathrm{hi}(v)).
\]
After computing $N(v)$ for all nodes bottom-up, a uniformly random strategy is obtained by starting at the root and, at each node $v$, choosing the lo-arc with probability $N(\mathrm{lo}(v))/N(v)$ and the hi-arc with probability $N(\mathrm{hi}(v))/N(v)$.
Whenever the hi-arc is chosen, we include the corresponding item $e_{\ell(v)}$.
This produces an exactly uniform draw from $\mathcal S$.

\paragraph{Length-stratified sampling.}
In many graph families, the distribution of strategy sizes $|S|$ is highly skewed, so uniform sampling over $\mathcal S$ can overwhelmingly produce long, dense strategies.
To control the sampled lengths, it is convenient to first choose a target length and then sample uniformly conditional on that length.

Let $N_r(v)$ denote the number of completions from $v$ that select exactly $r$ additional items.
These length-refined counts satisfy $N_0(\top)=1$, $N_r(\top)=0$ for $r\ge 1$, $N_r(\bot)=0$, and
\[
  N_r(v)=N_r(\mathrm{lo}(v)) + N_{r-1}(\mathrm{hi}(v)),
  \qquad r\ge 0,
\]
with the convention $N_{-1}(\cdot)=0$.
Given a desired remaining length $r$, a uniformly random strategy of that length is obtained by traversing from the root and selecting the lo-arc with probability $N_r(\mathrm{lo}(v))/N_r(v)$ (keeping $r$ unchanged) and selecting the hi-arc with probability $N_{r-1}(\mathrm{hi}(v))/N_r(v)$ (including $e_{\ell(v)}$ and decrementing $r$).

\paragraph{Choosing the length distribution.}
Let $\mathcal L:=\{r\ge 0 : N_r(\mathrm{root})>0\}$ denote the set of feasible lengths.
A general length-stratified sampler is defined by choosing a distribution $\pi$ on $\mathcal L$ and then sampling uniformly conditional on the selected length.
We use three choices:
\begin{itemize}
  \item \textbf{Uniform strategy (US).} Choose $r$ with probability $\pi(r)=N_r(\mathrm{root})/\sum_{\ell\in\mathcal L}N_{\ell}(\mathrm{root})$. This recovers uniform sampling over all strategies.
  \item \textbf{Uniform length (UL).} Choose $r$ uniformly over feasible lengths, i.e., $\pi(r)=1/|\mathcal L|$ for $r\in\mathcal L$. This equalizes the mass assigned to each length class.
  \item \textbf{Harmonic length (HL).} Choose $r$ with probability proportional to $1/r$ over $\mathcal L\cap\{1,2,\dots\}$. This further biases sampling toward shorter strategies.
\end{itemize}

\paragraph{Complexity and numerical stability.}
Once the diagram is built, uniform-strategy sampling requires only the node counts $N(v)$ and then one root-to-terminal traversal per sample.
Length-stratified sampling requires computing the table $(N_r(v))_{v,r}$ up to the maximum feasible length, after which each sample is generated by selecting a length and performing a single traversal with the appropriate conditional probabilities.
In practice the counts can be astronomically large, so in our Python implementation, we store them in log-scale and to form the branching probabilities via log-sum-exp arithmetic.
}{}

\section{Complexity of the Linear Minimization Oracle}\label{app:lmo_complexity}

We summarize the computational costs of the lower-level LMO used by Frank--Wolfe (\Cref{alg:fw-inner}), separating one-time preprocessing from the per-iteration work.
As in the main text, we write $n$ for the universe size (e.g., number of edges/resources), and denote by $N_{\mathrm{ZDD}}$ the number of internal nodes in the ZDD representing the strategy family $\mathcal S$ (see \Cref{app:zdd}).
For subsampling LMOs, $m$ denotes the number of sampled strategies per FW step. We also define $\ell_{\max} := \max_{S\in\mathcal S} |S|$ as the maximum strategy length (i.e., the maximum number of selected universe items), so that $\ell_{\max}\le n$. \Cref{tab:lmo-complexity} summarizes the resulting preprocessing and per-iteration costs.

\begin{table}[t]
\centering
\caption{Asymptotic costs of LMO implementations. Here $\ell_{\max}\le n$ upper-bounds the strategy length and is the natural worst-case cost of a single root-to-terminal traversal in the ZDD. Uniform-strategy (US) sampling only requires scalar node counts, while length-stratified schemes (UL/HL) require length-refined counts up to $\ell_{\max}$ (see \Cref{app:zdd_sampling} for more details).}
\label{tab:lmo-complexity}
\resizebox{\linewidth}{!}{%
\begin{tabular}{lll}
\toprule
LMO & Preprocessing (one-time) & Per FW iteration \\
\midrule
Poly-time family & none & $\mathrm{poly}(n)$ \\
ZDD exact & build ZDD: $\O(N_{\mathrm{ZDD}})$ & min-cost DP: $\O(N_{\mathrm{ZDD}})$ \\
ZDD subsampling (US) & build ZDD + counts: $\O(N_{\mathrm{ZDD}})$ & $m$ samples + scoring: $\O(m\,\ell_{\max})\subseteq \O(mn)$ \\
ZDD subsampling (UL/HL) & build ZDD + length counts: $\O(N_{\mathrm{ZDD}}\,\ell_{\max})$ & $m$ samples + scoring: $\O(m\,\ell_{\max})\subseteq \O(mn)$ \\
\bottomrule
\end{tabular}%
}
\end{table}

\paragraph{Poly-time strategy families.}
When the strategy family admits a polynomial-time combinatorial algorithm, the LMO can be implemented without preprocessing.
For example, for $s$--$t$ paths and additive edge weights, the LMO is a shortest-path problem, solvable in $\O(|E|\log|V|)$ time using Dijkstra's algorithm, where $|V|$ and $|E| = n$ are the numbers of vertices and edges in the graph.

\paragraph{NP-hard families with exact ZDD minimization.}
For NP-hard families, the ZDD provides an implicit representation of $\mathcal S$ that enables exact minimization of additive costs by dynamic programming on the ZDD DAG.
Once the diagram is built, a min-cost oracle call runs in time linear in the ZDD size, $\O(N_{\mathrm{ZDD}})$. 
The worst-case $N_{\mathrm{ZDD}}$ can be exponential in $n$, but in many structured instances it is far smaller than $|\mathcal S|$; our experiments illustrate this gap across three real transportation networks.

\paragraph{NP-hard families with subsampling.}
When $N_{\mathrm{ZDD}}$ is too large for repeated exact minimization, we still build the ZDD once and use it as a sampler.
Uniform-strategy sampling requires only node counts, while length-stratified sampling additionally uses length-refined counts up to $\ell_{\max}$.
Each FW step then draws $m$ feasible strategies and returns the best among the samples with respect to the current linearization.
This reduces the per-iteration cost to $\O(m\,\ell_{\max})$ at the expense of replacing the exact minimizer with a randomized approximation.

\section{Orders of Smoothness and Lipschitz Constants}\label{app:constant_orders}

This appendix relates the abstract constants in \Cref{ass:LL_UL_regular} to the primitives
$c_i(\cdot;\theta)$, and specializes the bounds to the fractional cost model used in
\Cref{sec:experiments} and Appendix~\ref{app:exp_details}.

\subsection{Smoothness of the Beckmann Potential}

Recall the Beckmann potential
\[
  f(\theta,y)=\sum_{i=1}^n\int_0^{y_i} c_i(u;\theta)\,du,
  \qquad
  \nabla_y f(\theta,y)=\bigl(c_i(y_i;\theta)\bigr)_{i=1}^n.
\]
Hence, the constant $L_{f,2}$ in \Cref{ass:LL_UL_regular} is exactly a Lipschitz constant of the
vector cost map $(\theta,y)\mapsto c(y;\theta)$ on $\Theta\times\mathcal C$.

A common sufficient condition is coordinatewise Lipschitzness: assume there exist constants
$L_y,L_\theta\ge 0$ such that for all $i\in[n]$ and all $(\theta,y),(\theta',y')\in\Theta\times\mathcal C$,
\[
  |c_i(y_i;\theta)-c_i(y_i';\theta')|
  \;\le\;
  L_y\,|y_i-y_i'| \;+\; L_\theta\,\|\theta-\theta'\|.
\]
Then
\[
  \|\nabla_y f(\theta,y)-\nabla_y f(\theta',y')\|
  \;=\;
  \|c(y;\theta)-c(y';\theta')\|
  \;\le\;
  L_\theta\|\theta-\theta'\| + L_y\|y-y'\|,
\]
so \Cref{ass:LL_UL_regular} holds with $L_{f,2}:=\max\{L_y,L_\theta\}$.
In particular, for each fixed $\theta$, the map $y\mapsto f(\theta,y)$ is $L_{f,2}$-smooth on $\mathcal C$.

\subsection{Lipschitzness of the Leader Objective}

For the social-cost objective used in the experiments,
\[
  F(\theta,y)=\sum_{i=1}^n y_i\,c_i(y_i;\theta_i),
\]
a convenient way to certify \Cref{ass:LL_UL_regular} is via gradient bounds.
Assume $F$ is differentiable on $\Theta\times\mathcal C$ and define the uniform coordinate bounds
\[
B_y:=\sup_{(\theta,y)\in\Theta\times\mathcal C}\max_{i\in[n]}\bigl|\partial_{y_i}F(\theta,y)\bigr|,
\qquad
B_\theta:=\sup_{(\theta,y)\in\Theta\times\mathcal C}\max_{j\in[k]}\bigl|\partial_{\theta_j}F(\theta,y)\bigr|.
\]
Then $\|\nabla_y F(\theta,y)\|\le \sqrt n\,B_y$ and $\|\nabla_\theta F(\theta,y)\|\le \sqrt k\,B_\theta$, hence for all
$(\theta,y),(\theta',y')\in\Theta\times\mathcal C$,
\[
  |F(\theta,y)-F(\theta',y')|
  \;\le\;
  \sqrt k\,B_\theta\,\|\theta-\theta'\| + \sqrt n\,B_y\,\|y-y'\|
  \;\le\;
  L_{F,1}\bigl(\|\theta-\theta'\|+\|y-y'\|\bigr),
\]
where one may take $L_{F,1}:=\max\{\sqrt k\,B_\theta,\sqrt n\,B_y\}$.
In particular, if the partial derivatives are uniformly $\O(1)$, then typically $L_{F,1}=\O(\sqrt{n}+\sqrt{k})$.

\subsection{Specialization to the Fractional Cost Model}

In \Cref{sec:experiments} we use
\[
  c_i(y_i;\theta_i)=d_i\Bigl(1 + C\,\frac{y_i}{\theta_i+1}\Bigr),
  \qquad
  \Theta=\{\theta\in\R_+^n:\mathbf 1^\top\theta=n\},
  \qquad
  y\in\mathcal C\subseteq[0,1]^n,
\]
and we normalize $d_i\le 1$.

\paragraph{Bound on $L_{f,2}$.}
For all $i$ and all $(\theta,y)$ in the domain,
\[
  \frac{\partial}{\partial y_i}c_i(y_i;\theta_i)= d_i\,\frac{C}{\theta_i+1}\le C,
  \qquad
  \Bigl|\frac{\partial}{\partial \theta_i}c_i(y_i;\theta_i)\Bigr|
  = d_i\,C\,\frac{y_i}{(\theta_i+1)^2}\le C.
\]
Therefore $(\theta,y)\mapsto c(y;\theta)$ is $C$-Lipschitz in $y$ and also $C$-Lipschitz in $\theta$
(on $\Theta\times\mathcal C$ under the Euclidean norm), so \Cref{ass:LL_UL_regular} holds with
\[
  L_{f,2}=\O(C).
\]

\paragraph{Bound on $L_{F,1}$.}
Writing
\[
  F(\theta,y)=\sum_{i=1}^n \Bigl(d_i y_i + d_i C\,\frac{y_i^2}{\theta_i+1}\Bigr),
\]
we have, for each $i$,
\[
  \frac{\partial F}{\partial y_i}
  = d_i + 2d_iC\,\frac{y_i}{\theta_i+1}
  \le 1+2C,
  \qquad
  \Bigl|\frac{\partial F}{\partial \theta_i}\Bigr|
  = d_iC\,\frac{y_i^2}{(\theta_i+1)^2}
  \le C.
\]
Hence $\|\nabla_y F(\theta,y)\|\le \sqrt{n}\,(1+2C)$ and $\|\nabla_\theta F(\theta,y)\|\le \sqrt{n}\,C$,
so one may take
\[
  L_{F,1}=\O\bigl(\sqrt{n}\,(1+C)\bigr).
\]

\paragraph{Implication for $L_\Phi$.}
Combining with Lemma~\ref{lem_solution_map_Lip}, a valid choice is
\[
  L_\Phi
  = L_{F,1}\Bigl(1+\frac{L_{f,2}}{\alpha}\Bigr)
  = \O\!\Bigl(\sqrt{n}\,(1+C)\Bigl(1+\frac{C}{\alpha}\Bigr)\Bigr).
\]

\end{document}